\documentclass[11pt]{article}

\pdfoutput=1
\usepackage{amsmath,amssymb}
\usepackage{algorithm,algorithmic}
\usepackage{amsthm}
\usepackage{dsfont}
\usepackage{fullpage}
\usepackage{graphicx}
\usepackage{subfigure}

\newtheorem{lemma}{Lemma}
\newtheorem{theorem}{Theorem}
\newtheorem{prop}{Proposition}
\newtheorem{corollary}{Corollary}

\newcommand{\defas}{\colon =}
\newcommand{\grad}{\bigtriangledown}

\newcommand{\argmax}{\operatorname{argmax}}

\newcommand{\U}{\mathcal{U}}
\newcommand{\V}{\mathcal{V}}
\newcommand{\G}{\mathcal{G}}
\newcommand{\B}{\mathcal{B}}
\newcommand{\mS}{\mathcal{S}}
\newcommand{\mL}{\mathcal{L}}

\author{Praneeth Netrapalli\\ praneethn@utexas.edu
\and Sujay Sanghavi \\ sanghavi@mail.utexas.edu
}

\date{}

\begin{document}

\title{Finding the Graph of Epidemic Cascades}

\maketitle
\begin{abstract}
We consider the problem of finding the graph on which an epidemic cascade spreads, given {\em only} the times when each node gets infected. While this is a problem of importance in several contexts -- offline and online social networks, e-commerce, epidemiology, vulnerabilities in infrastructure networks -- there has been very little work, analytical or empirical, on finding the graph.  Clearly, it is impossible to do so from just one cascade; our interest is in learning the graph from a small number of cascades. 

For the classic and popular ``independent cascade'' SIR epidemics, we analytically establish the number of cascades required by both the global maximum-likelihood (ML) estimator, and a natural greedy algorithm. Both results are based on a key observation: the global graph learning problem decouples into $n$ local problems -- one for each node. For a node of degree $d$, we show that its neighborhood can be reliably found once it has been infected $O(d^2 \log n)$ times (for ML on general graphs) or $O(d\log n)$ times (for greedy on trees). We also provide a corresponding information-theoretic lower bound of $\Omega(d\log n)$; thus our bounds are essentially tight. Furthermore, if we are given side-information in the form of a super-graph of the actual graph (as is often the case), then the number of cascade samples required -- in all cases -- becomes independent of the network size $n$. 

Finally, we show that for a very general SIR epidemic cascade model, the Markov graph of infection times is obtained via the moralization of the network graph. 

\end{abstract}

\section{Introduction}

Cascading, or epidemic, processes are those where the actions, infections or failure of certain nodes increase the susceptibility of other nodes to the same; this results in the successive spread of infections / failures / other phenomena from a small set of initial nodes to a much larger set. Initially developed as a way to study human disease propagation, cascade or epidemic processes have recently emerged as popular and useful models in a wide range of application areas. Examples include \\
{\em (a) social networks:} cascading processes provide natural models for understanding both the consumption of online media (e.g. viral videos, news articles\cite{LerGho10}) and spread of ideas and opinions (e.g. trending of topics and hashtags on Twitter/Facebook\cite{Zhou10}, keywords on blog networks\cite{Gruhletal04}) \\
{\em (b) e-commerce:} understanding epidemic cascades (and, in this case, finding influential nodes) is crucial to viral marketing \cite{KemKleTar03}, and predicting/optimizing uptake on social buying sites like Groupon etc. \\
{\em (c) security and reliability:} epidemic cascades model both the spread of computer worms and malware \cite{KepWhi91}, and cascading failures in infrastructure networks \cite{Kosetal99,Sacetal00} and complex organizations \cite{Per99}. \\
{\em (d) peer-to-peer networks:} epidemic protocols, where users sending and receiving (pieces of) files in a random uncoordinated fashion, form the basis for 
many popular peer-to-peer content distribution, caching and streaming networks \cite{MasTwi08,Bonetal08}. 

{\bf Structure Learning:} The vast majority of work on cascading processes has focused on understanding how the graph structure of the network (e.g. power laws, small world, expansion etc.) affects the spread of cascades. We focus on the {\em inverse problem:} if we only observe the states of nodes as the cascades spread, can we infer the underlying graph ? Structure learning is the crucial first step before we can {\em use} network structure; for example, before we find influential nodes in a network (e.g. for viral marketing) we need to know the graph. Often however we may only have crude, prior information about what the graph is, or indeed no information at all. \\
For example, in online social networks like Twitter or Facebook, we may have access to a {\em nominal} graph of all the friends of a user. However, clearly not all of them have an equal effect on the user's behavior; we would like to find the sub-graph of important links. In several other settings, we may have no a-priori information; examples include information forensics that study the spread of worms, and offline settings like real-world epidemiology and social science. The standard practice seems to be to use crude/nominal subgraphs if they exist (e.g. Twitter), or find graphs by other means (e.g. surveys). We propose to take a {\em data-driven} approach, finding graphs from observations of the cascades themselves.

While structure learning from cascades is an important primitive, there has been very little work investigating it (we summarize below). There are two related issues that need to be addressed: {\em (a) algorithms:} what is the method, and its complexity, and {\em (b) performance:} how many observations are needed for reliable graph recovery? The main intellectual contribution of this paper is characterizing the performance of two algorithms we develop, and a lower bound showing they perform close to optimal. To the best of our knowledge, there exists no prior work on performance analysis (i.e. characterizing the number of observations needed) for learning graphs of epidemic cascades.

\subsection{Summary of Our Results}

We present two algorithms, and information-theoretic lower bounds, for the problem of learning the graph of an epidemic cascade when we are given prior information of a super-graph\footnote{Of course if no super-graph is given, it can be taken to be the complete graph.}. It is not possible to learn the graph from a single cascade; we study the number of cascades required for reliable learning. Key outcomes of our results are that {\em (i)} epidemic graph learning can be done in a fast, distributed fashion, {\em (ii)} with a number of samples that is close to the lower bound. Our results:

{\bf (a) Maximum Likelihood:} We show that, via a suitable change of variables, the problem of finding the graph most likely to generate the cascades we observe {\em decouples} into $n$ convex problems -- one for each node. Further, for node $i$, the algorithm requires as input only the infection times of that node's size-$D_i$ super-neighborhood; it is local both in computation and in the information requirement. Our main result here is to establish that for this efficient algorithm, if $d_i$ is the size of the true neighborhood, then node $i$ needs to be infected  $O(d_i^2\log D_i)$ times before we learn it, for a general graph. 

{\bf (b) Greedy algorithm:} We show that if the graph is a tree, then a natural greedy algorithm is able to find the true neighborhood of a node $i$ with only $O(d_i\log D_i)$ samples. The greedy algorithm involves iteratively adding to the neighborhood the node which ``explains" (i.e. could be the likely cause of) the largest number of instances when node $i$ was infected, and removing those infections from further consideration. 

{\bf (c) Lower bounds:} We first establish a general information-theoretic lower bound on the number of cascade samples  required for approximate graph recovery, for general (but abstract) notions of approximation, and for any SIR process. We then derive two corollaries: one for learning a graph upto a specified edit distance when there is no super-graph information, and another for the case when there is a super-graph, and specified edit distances for each of the nodes. These bounds show that the ML algorithm is at most a factor $d$ away from the optimal. 

{\bf (d) Markov structure of general cascades:} Every set of random variables has an associated Markov graph. In our final result, we show  that for a very general SIR epidemic cascade model -- essentially any that is causal with respect to time and the directed network graph -- the (undirected) Markov graph of the (random) infection times is the {\em moralized} graph of the true directed network graph on which the epidemic spreads. This allows for learning graph structure using techniques from Markov Random Fields / graphical models, and also illustrates the role of causality.

While here we used the $O(\cdot)$ and $\Omega(\cdot)$ notation for compact statement, we emphasize that our results are {\em non-asymptotic}, and thus more general than a merely asymptotic result. Thus for fixed values of system parameters and probabilities of error, we give precise bounds on the number of cascades we need to observe. If one is interested in asymptotic results under particular scaling regimes for the parameters, such results can be derived as corollaries of our algorithms (with union bounds if one is interested in complete graph recovery).

A nice feature of our results is that both the algorithms work on a {\em node by node basis}. Thus for recovering the neighbors of a node we only need information about its super-neighborhood, and solve a local problem. We are also able to find the neighborhood of one or a few nodes, without worrying about finding the neighborhoods of other nodes or the entire graph. Similarly, the number of samples required to recover the neighborhood of a node depend only on the sizes of its own neighborhood and super-neighborhood. 

\subsection{Related Work}

\noindent {\bf Learning graphs of epidemic cascades:} While structure learning from cascades is an important primitive, there has been very little work investigating it: \\
{\em (a) algorithms:} A recent paper \cite{RodLesKra10} investigates learning graphs from infection times for the independent cascade model (similar setting as our paper). However, they take an approach that results in an NP-hard combinatorial optimization problem, which they show can be approximated. Another paper \cite{MyeLes10} shows max-likelihood estimation in the independent cascade model can be cast as a decoupled convex optimization problem (albeit a different one from ours). \\
{\em (b) performance:} To the best of our knowledge, there has been no work on the crucial question of how many cascades one needs to observe to learn the graph; indeed, this question is the main focus of our paper.

\noindent {\bf Markov graph structure learning:} The ideas in this paper are related to those from Markov Random Fields (MRFs, aka Graphical Models) in statistics and machine learning, but there are also important differences. We overview the related work, and contrast it to ours, in Section \ref{sec:markov}.

\section{System Model}\label{sec:sysmodel}

Most of the analytical results of this paper are for the classic and popular {\em independent cascade} model of epidemics; in particular we will consider the simple one-step model first proposed in \cite{GolLibMul01} and recently popularized by Kempe, Kleinberg and Tardos \cite{KemKleTar03}. 

{\bf Standard independent cascade epidemic model \cite{KemKleTar03}:} The network is assumed to be a {\em directed} graph $G=(V,E)$;  for every directed edge $(i,j)$ we say $i$ is a parent and $j$ is a child of the corresponding other node. Parent may infect child along an edge, but the reverse cannot happen; we allow bi-directed edges (i.e. it is possible that $(i,j)$ and $(j,i)$ are in $E$). Let $\V_i := \{j : (j,i)\in E\}$ denote the set of parents of each node $i$, and for convenience we also include $i\in \V_i$. Epidemics proceed in discrete time; all nodes are initially in the {\em susceptible} state. At time 0, each node tosses a coin and independently becomes {\em active}, with probability $p_{init}$. This set of initially active nodes are called {\em seeds}. In every time step each active node probabilistically tries to infect its susceptible children; if node	 $i$ is active at time $t$, it will infect each susceptible child $j$ with probability $p_{ij}$, independently. Correspondingly, a node $j$ that is susceptible at time $t$ will become active in the next time step, i.e. $t+1$, if {\em any one} of its parents infects it. Finally, a node remains active for only {\em one} time slot, after which it becomes {\em inactive:} it does not spread the infection, and cannot be infected again. Thus this is an ``SIR" epidemic, where some nodes remain forever susceptible because the epidemic never reaches them, while others transition according to: \\ $\quad \quad$ {\bf susceptible $\rightarrow$ active for one time step $\rightarrow$ inactive.}
A sample path of the independent cascade model is illustrated in Figure \ref{fig:indcascade-evolution}.

\begin{figure}[h]
\centering
\subfigure{
\includegraphics[width=0.22\linewidth]{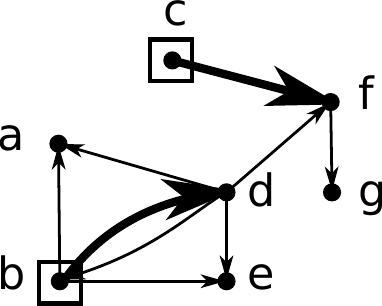}
\label{fig:indcascade-t0}
}
\subfigure{
\includegraphics[width=0.22\linewidth]{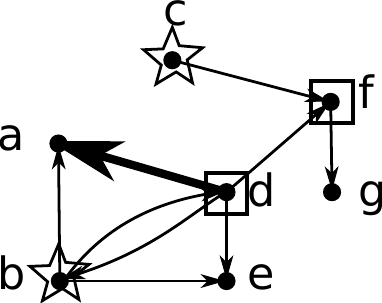}
\label{fig:indcascade-t1}
}
\subfigure{
\includegraphics[width=0.22\linewidth]{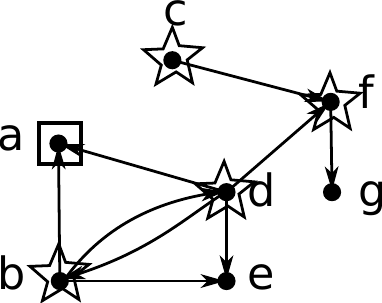}
\label{fig:indcascade-t2}
}
\subfigure{
\includegraphics[width=0.22\linewidth]{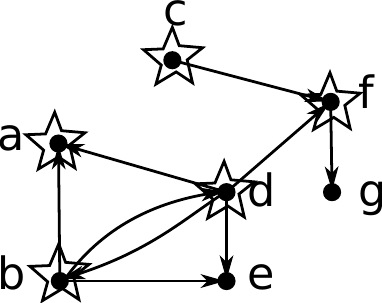}
\label{fig:indcascade-t3}
}
\caption{{\bf Illustration of the independent cascade model:} 
This figure illustrates a sample path of the evolution of the independent cascade model.
The four figures above represent the state of the system at time steps $0$, $1$, $2$ and $3$
respectively. A node with no box around it means that it is in susceptible state, a node with
a square around it means that it is active and a node with a star around it means that it is
inactive. At time step $0$, nodes $b$ and $c$ are chosen as seeds. They infect
$d$ and $f$ respectively and turn inactive. In time step $1$, $d$ infects $a$ where as $f$ fails
to infect any of its children. In time step $2$, $a$ does not have any children to infect. Once
$a$ turns inactive in time step $3$, the epidemic stops.}
\label{fig:indcascade-evolution}
\end{figure}

Note thus that the set parental set is $\V_i = \{j : p_{ji} >0\}$, i.e. the set of all nodes that have a non-zero probability of infecting $i$.

{\bf Observation model:} For an epidemic cascade $u$ that spreads over a graph, we observe for each node $i$ the time $t_i^u$ when $i$ became active. If $i$ is one of the seed nodes of cascade $u$ then $t^u_i=0$, and for nodes that are never infected in $u$ we set $t^u_i = \infty$. Let $t^u$ denote the vector of infection times for cascade $u$. We observe more than one cascade on the same graph; let $\U$ be the set of cascades, and $m = |\U|$ be the number, which we will often refer to as the {\em sample complexity}. Each cascade is assumed to be generated and observed as above, independent of all others.

{\bf (possible) Super-graph information:} In several applications, we (may) also have prior knowledge about the network, in the form of a directed {\em super-graph}\footnote{For example, on social networks like Facebook or Twitter, we may know the set of all friends of a user, and from these we want to find the ones that most influence the user.} of $G$. We find it convenient to represent super-graph information as follows: for each node $i$, we are given a set $\mS_i \subset V$ of nodes that contain its true parents; i.e. $\V_i \subset \mS_i$ for all $i$. In terms of edge probabilities, this means that $p_{ji} >0$ (strictly) for $j\in \V_i$, and $p_{ji} = 0$ for $j\in \mS_i \backslash \V_i$. Of course if no super-graph is available we can set $\mS_i = V$, the set of all nodes; so from now on we assume a $\mS_i$ is always available.

{\bf Problem description:} Using the vectors of infection times $\{t^u\}$ we are interested in finding the parental neighborhood $\V_i$, for some or all of the nodes $i$. That is, we want to find the set of nodes that can infect $i$. This is not possible when we only observe a single cascade; we will thus be interested in learning the graph from {\em as few cascades as possible.} 

Note that multiple seeds begin each cascade $u\in \U$; thus, for a single cascade even at time step 1 we will not be able to say with surety which seed infected which individual. 

{\bf Correlation decay:} Loosely speaking, random processes on graphs are said to have ``correlation decay" if far away nodes have negligible effects. For our problem, this means that the cascade from each seed does not travel too far. Formally, all the results in this paper assume that there exists a number $\alpha >0$ such that for every node $i$, the sum of all probabilities of incoming edges satisfies $\sum_k p_{ki} < 1 - \alpha$. The following lemma clarifies what this assumption means for the infection times of a node.
\begin{lemma}\label{lem:corrdecay-probbound}
For any node $i$ and time $t$, we have
\begin{align*}
  \mathbb{P}\left[T_i = t \right] \leq \left(1-\alpha\right)^{t-1} p_{\textrm{init}}
\end{align*}
\end{lemma}
Thus, the probability $\mathbb{P}[T_i < \infty]$  that a node is infected satisfies $p_{init} < \mathbb{P}[T_i < \infty] < \frac{p_{init}}{\alpha}$. Also, the average distance from a node to any seed that infected it is at most $\frac{1}{\alpha}$. We discuss the case where there is no correlation decay in the Discussion section.

{\bf Interpreting the results:} Each cascade we observe provides some information about the graph. Suppose we want to infer the presence, or absence, of the directed edge $(i,j)$ (i.e. if $p_{ij} >0$ or not). Note that if the parent $i$ is not infected in a cascade, then {\em that cascade provides no information} about $(i,j)$: since the parent was never infected, no infection attempt was made using that edge; the ``edge activation variable" was never sampled. While our theorems are in terms of the total number $m$ of cascades needed for graph estimation, for a meaningful interpretation of this number one needs to realize that the expected number of times we get {\em useful} information about any edge is, on average, between $m p_{init}$ and $m p_{init}/\alpha$. These are also the bounds on the average number of times a particular node is infected in a particular cascade.

We provide both upper bounds (via two learning algorithms), and (information theoretic) lower bounds on the sample complexity. Note that the {\em execution} of our algorithms does not require knowledge of these parameters like $p_{init}, \alpha$ etc.; these are defined only for the {\em analysis}.

\section{Maximum Likelihood}

The graph learning problem can be interpreted as a parameter estimation problem: for each cascade, the vector $T$ of infection times  is a set of random variables that has a joint distribution which is determined by a set of parameters $p_{ji} \geq 0$ for every $i$ and $j\in \mS_i$. We want to find these parameters, or more specifically the identities of the edges where they are non-zero, from samples $t^u$, $u\in \U$. Each choice of  parameters has an associated probability, or likelihood, of generating the infection times we observe. The classical {\em Maximum-likelihood (ML) estimator} advocates picking the parameter values that maximize this likelihood.

Our crucial insight in this section is that, with an appropriate change of variables the likelihood function has a particularly nice (decoupled, convex) form, enabling both efficient implementation and analysis. In particular, define  $\theta_{ij} := - \log (1-p_{ij})$ ; note that  $p_{ij} = 0 \Leftrightarrow \theta_{ij} = 0$. 

Further, for each node $i$ let $\theta_{*i} := \{\theta_{ji}\,;\,j\in \mS_i\}$ be the set of parameters corresponding to the possible parents $\mS_i$ of node $i$. Let $\theta$ be the set of all parameters of the graph. Note that $\theta \geq 0$ (i.e. every parameter is positive or zero). Finally, we define the {\em log-likelihood} of a vector $t$ of samples to be 
\[
\mathcal{L}(t;\theta) ~ := ~ \log \left ( \mathrm{Pr}_\theta[T=t] \right )
\]
The proposition below shows how $\mL$ decouples into convex functions with this change of variables.

\begin{prop}[convexity \& decoupling]\label{prop:likelihood} 
For any vector of parameters $\theta$, and infection time vector $t$, the log-likelihood is given by
\[
\mathcal{L}(t;\theta) ~ = ~ \log (p_{init}^s (1-p_{init})^{n-s}) \, + \sum_i \mathcal{L}_i(t_{\mS_i};\theta_{*i}) 
\]
where $s$ is the number of seeds (i.e. nodes with $t_i = 0$), and the node-based term
\[
\mathcal{L}_i(t_{\mS_i};\theta_{*i}) ~ := ~ - \sum_{j:t_j\leq t_i-2} \theta_{ji} \, + \, \log \left ( 1 - \exp \left ( -\sum_{j:t_j = t_i -1} \theta_{ji} \right ) \right ) 
\]
Furthermore, $\mathcal{L}_i(t_{\mS_i};\theta_{*i})$ is a concave function of $\theta_{*i}$, for any fixed $t_{\mS_i}$.
\end{prop}

{\bf Proof:} Please see appendix. 

{\bf Remark:} The overall log-likelihood $\mathcal{L}(t;\theta)$ has now decoupled because it is the sum of $n$ terms of the form $\mathcal{L}_i(t;\theta_{*i})$, each of which depend on a {\em different} set of variables $\theta_{*i}$. Thus each one can be optimized, and analyzed, in isolation.

The {\em algorithmic} implications of this proposition are: \\
{\em (a)} if we are only interested in a small subset of nodes, we can find their parental neighborhood by solving a separate $|\mS_i|$-variable convex program for each one, \\
{\em (b)} even if we want to find the entire graph, the decoupling allows for parallelization, and speedup: solving $n$ convex programs with $n$ variables each is much faster than solving one program with $n^2$ variables. \\
{\em (c)} The function $\mathcal{L}_i$ is fully determined by the times $t_{\mS_i}$ of the node's super-neighborhood; it does not need knowledge of the infection times of other nodes.

Proposition \ref{prop:likelihood} is equally crucial {\em analytically,} as it enables us to derive bounds on the number of cascades required for us to reliably select the neighborhood, via analysis of the first-order optimality conditions of the convex program. In particular, we will see that complementary slackness conditions from convex programming, and concentration results, are key to proving our results on the sample complexity of the ML procedure.

The ML algorithm for finding the parental neighborhood of node $i$ is formally stated below. it involves solving the convex program corresponding to the max-likelihood, and setting small values of $\theta_{ji}$ to 0. The threshold for this cut-off is $\eta$, which is an input to the procedure.

\begin{algorithm}[h]
\caption{ML Algorithm for Node $i$}
\label{algo:ML}
\begin{algorithmic}[1]
\STATE Find the optimizer of the empirical likelihood, i.e. find
\[
\widehat{\theta}_{*i} ~ := ~ \arg \max_{\theta_{*i}} \,  \sum_u \, \mathcal{L}_i(t^u_{\mS_i};\theta_{*i})
\]
where $\mathcal{L}_i(t_{\mS_i};\theta_{*i})$ is as defined in Prop. \ref{prop:likelihood}.
\STATE Estimate the parental neighborhood by thresholding: 
\[
\widehat{\V}_i ~ := ~ \{ j : \widehat{\theta}_{ji} \geq \eta \}
\]
\STATE Output $\widehat{\V}_i$.
\end{algorithmic}
\end{algorithm}

Our main analytical result of this section is a characterization of the performance of this ML algorithm, in terms of the number of cascades it needs to reliably estimate the parental neighborhood of any node $i$.

\begin{theorem}\label{thm:samplecomplexity_maxLL}
Consider a node $i$ with true parental degree $d_i := |\V_i|$, and super-graph degree $D_i := |\mS_i|$. Let $p_{i,min} := \min_{j\in \V_i} p_{ji}$ be the strength of the edge from the weakest parent. Assume $d_i p_{init} < \frac{1}{2}$. Then, for any $\delta >0$, if the number of cascades $m = |\U|$ satisfies
\begin{equation}
m ~ > ~ \frac{c}{p_{init}} \left ( \frac{1}{\alpha^7 \eta^2  p_{i,min}^2} \right )  d_i^2 \log \left(\frac{D_i}{\delta}\right) \label{eq:ML_sampcomp}
\end{equation}
Then, with probability greater than $1-\delta$, the estimate $\widehat{\V}_i$ from the ML algorithm with threshold $\eta$ will have \\
{\em (a)} no false neighbors, i.e. $\hat{V}_i \subset \V_i$, and \\
{\em (b)} all strong enough neighbors: if $j\in \V_i$ and $p_{ji} > \frac{8}{\alpha} (e^{2\eta} - 1)$, then $j\in \widehat{\V}_i$ as well. \\
Here $c$ is a number independent of any other system parameter.
\end{theorem}

\noindent {\bf Remarks:} 

{\em (a)} This is a {\em non-asymptotic} result that holds for {\em all} values of the system variables $d_i,p_{init},\alpha,p_{i,min},\eta$ and $\delta$. Appropriate asymptotic results can be derived as corollaries, if required. Note that this result on finding the nodes that influence node $i$ {\em does not depend on $n$}.

{\em (b)} We can learn the entire neighborhood, i.e. $\widehat{\V}_i = \V_i$, by choosing the threshold $\eta \leq \frac{1}{2} \log (1+\frac{\alpha p_{i,min}}{8})$ low enough, and the corresponding number of cascade samples $m$ according to (\ref{eq:ML_sampcomp}). Thus, the number of times node $i$ needs to be infected before we can reliably (i.e. with a fixed small error probability) learn its neighborhood scales as $O(d_i^2 \log D_i)$ (for fixed values of other system variables). Our result allows for learning stronger edges with fewer samples.

{\em (c)} If we want to learn the structure of the {\em entire} graph with probability greater than $\epsilon$, we can set $\delta = \epsilon / n$ and then take a union bound over all the nodes. So, for example, if every node has true degree at most $|\V_i| \leq d$, and super-graph degree $|\mS_i| \leq D$, then the number of samples needed to learn the entire graph (with probability at least $1-\epsilon$) scales as $O(d^2 \log \frac{Dn}{\epsilon})$ (for fixed values of other system variables).

%{\em (d)} {\bf @@ reasoning behind pmin} \\

{\em (d)} The average number of parents of $i$ that are seeds is $d_i p_{init}$. If this is large, then in every cascade there will be a reasonable probability of one of them being seeds, and infecting $i$ in the next time slot. This makes it hard to discern the neighborhood of $i$; the (mild) assumption $d_i p_{init} < \frac{1}{2}$ is required to counter this effect. Indeed, in most applications $p_{init}$ is likely to be quite small.

{\em (e)} Note that our results depend on the {\em in-degree} of nodes, not the out-degree. So for example it is possible to have high out-degree nodes (as e.g. in power-law graphs), and still be able to learn the graph with small number of samples. 

\subsection{Generalized Independent Cascade Model}

In this paper, for ease of analysis, we restrict our sample-complexity analysis to one-step independent cascade epidemics, where a node is active for only one time slot after it is infected. However, our algorithmic and bounding approaches apply to a more general class of independent cascade models. Specifically, we consider an extension where each parent now has a probability distribution of the amount of time it waits before infecting a child, and prove a generalization of Proposition \ref{prop:likelihood}, which was the key result enabling both the implementation and analysis of the ML algorithm.

Formally, let $p_{ji}^{\tau}$ denote the probability that an active node $j$ infects a susceptible child $i$, $\tau$ time steps
after $j$ was infected. The time taken for $j$ to infect $i$ is bounded by a parameter $\overline{t}$ i.e., $p_{ji}^{\tau}=0$
for $\tau>\overline{t}$. Note that if we have $\overline{t} = 1$, we recover the standard
independent cascade model. The total probability that $j$ infects $i$ is given by $\sum_{\tau\in[\overline{t}]} p_{ji}^{\tau}$
(which can be strictly less than $1$) where $[\overline{t}]$ denotes the set of integers between $1$ and $\overline{t}$
(including the end points). 

Following in the steps of Proposition \ref{prop:likelihood}, define $\theta_{ji}^{\tau} \defas -\log\left(\frac{1-\sum_{r\in[\tau]} p_{ji}^r}{1-\sum_{r\in[\tau-1]} p_{ji}^r}\right)$.
Note that given any parameter vector $p_{ji}^{\tau}$ we obtain the corresponding $\theta_{ji}^{\tau}$ and vice versa.
Moreover $\theta_{ji}^{\tau}=0 \Leftrightarrow p_{ji}^{\tau}=0$.
Suppose each node is seeded with the infection with probability $p_{\textrm{init}}$ and let
$\mL(t,\theta)$ denote the log-likelihood of the infection time vector $t$ when the parameters of the model are given by
$\theta$. We have the following version of Proposition \ref{prop:likelihood} for the generalized independent cascade model.

\begin{prop}\label{prop:likelihood-genIC}
  For any vector of parameters $\theta$, and infection time vector $t$, the log-likelihood is given by
\[
\mathcal{L}(t;\theta) ~ = ~ \log (p_{init}^s (1-p_{init})^{n-s}) \, + \sum_i \mathcal{L}_i(t_{\mS_i};\theta_{*i}) 
\]
where $s$ is the number of seeds (i.e. nodes with $t_i = 0$), and the node-based term
\[
\mathcal{L}_i(t_{\mS_i};\theta_{*i}) ~ := ~ - \sum_{j:t_j\leq t_i-2} \sum_{\tau \in [t_i-t_j-1]} \theta_{ji}^{\tau} \,
	+ \, \log \left ( 1 - \exp \left ( -\sum_{j:t_j<t_i} \theta_{ji}^{t_i-t_j} \right ) \right ) 
\]
Furthermore, $\mathcal{L}_i(t_{\mS_i};\theta_{*i})$ is a concave function of $\theta_{*i}$, for any fixed $t_{\mS_i}$.
\end{prop}

{\bf Proof:} Please see appendix. 

\section{Greedy Algorithm}

We now analyze the sample complexity of a simple iterative greedy algorithm -- for the case when the graph is a tree\footnote{We believe (especially since we have correlation decay) that our results can be easily extended to the case of ``locally tree-like" graphs; e.g. random graphs from the Erdos-Renyi, random regular or several other popular models.}. The algorithm is of course defined for general graphs.

The idea is as follows: suppose we want to find the parents of node $i$ from a given set of cascades $\U$. In each cascade $u$, the set of nodes that could have possibly infected $i$ is the set of nodes $j$ for which $t^u_j = t^u_i -1$. In the first step, the algorithm thus picks the $j$ which has $t^u_j = t^u_i -1$ for the largest number of observed cascades. It then {\em removes} those cascades from further consideration (since they have been ``accounted for") and proceeds as before on the remaining cascades, stopping when all cascades are exhausted.

%
%Note that if an node $j$ infects $i$ in epidemic $u$, then $t^u_j = t^u_i -1$; of course several other nodes could also have the same infection time as $j$. In the first step, the greedy algorithm thus picks the node $j$ that has the highest number of epidemics where it is infected one time step before $i$. These epidemics are then deemed ``accounted for", and removed for subsequent steps. The algorithm then proceeds as before, iteratively picking the node that has the highest incidence of 1-time-step precedence to infection of $i$, and removing the corresponding epidemics from further consideration. .

\begin{algorithm}[h]
\caption{Greedy Algorithm for Node $i$}
\label{algo:greedy}
\begin{algorithmic}[1]
\STATE 	Initialize unaccounted cascades $U = \mathcal{U}$
\STATE Initialize $\widehat{\V}_i = \emptyset$
\WHILE{$U \neq \emptyset$}
\STATE Find $k =  \arg \max_{j \in \mS_i} |\{ u\in U : t^u_j = t^u_i -1 \}|$
\STATE Add it : $\widehat{\V}_i \leftarrow \widehat{\V}_i  \cup k$
\STATE Remove cascades : $U \leftarrow U \setminus \{u : t_k^u = t_i^u - 1\}$
\ENDWHILE
\STATE Output $\widehat{\V}_i$
\end{algorithmic}
\end{algorithm}

Our main result for this section is below.

\begin{theorem}\label{thm:greedy-guarantees}
Suppose the graph $G$ is a tree, and the degree of node $i$ is $d_i := |\V_i|$.
Suppose also that $p_{\textrm{init}} < \frac{\alpha^2 p_{\textrm{min}}}{16ed}$.
If Algorithm \ref{algo:greedy} is given a super-neighbhorhood of size $D_i := |\mS_i|$, then for any $\delta >0$ if the number of samples satisfies
\[
m ~ > ~ \frac{c}{p_{init}} \left ( \frac{1}{p_{min}} \right ) d_i \log \frac{D_i}{\delta}
\]
then with probability at least $1-\delta$  the estimate from the greedy algorithm will be the same as the true neighborhood, i.e.  $\widehat{\V}_i = \V_i$. Here $c$ is a constant independent of any other system parameter.

\end{theorem}

\section{Lower Bounds}

We now turn our attention to establishing lower bounds on the number of cascades that need to be observed for even approximately learning graph structure, using {\em any} algorithm. Clearly, we now cannot focus on learning just one graph, since in that case we could come up with an ``algorithm" tailored to find precisely that one graph. Instead, as is standard practice in information-theoretic lower bounds, we need to consider a collection (or ``ensemble") of graphs, and study how many cascades are needed to (approximately) find {\em any one} graph from this collection.  

We first state a lower bound in a general setting, for any pre-defined ensemble and notion of approximate recovery. We then provide two corollaries specializing it to our independent cascade epidemic model, edit distance approximation, and two natural graph ensembles.

{\bf General Setting:} Consider any general cascading process generating infection times $\{T_i\}$. Let $\G$ be a fixed collection of graphs and corresponding edge probabilities, and let $G$ be a graph chosen uniformly at random from this collection. We then generate a set $\U$, with $|\U| = m$, of independent cascades, and observe infection times $T^\U$. Let $\widehat{G}(T^\U)$ be a graph estimator that takes the observations as an input and outputs a graph. Finally, we say that a graph $G'$ approximately recovers graph $G$ if $G \in \B(G')$, where  $\B(G')\subseteq \G$ is any pre-defined set of graphs, with one such set defined for every $G'$. 

So for example, if we are interested in exact recovery, we would have $\B(G') = \{G'\}$, i.e. the singleton. If we were interested in edit distance of $s$, we would have $\B(G')$ be the set of all graphs within edit distance $s$ of $G'$.

We define the probability of error of a graph estimator $\widehat{G}(\cdot)$ to be
\[
P_e(\widehat{G}) ~ := ~ \mathbb{P}[G \notin \B(\widehat{G}(T^\U))]
\]
where the probability is calculated over the randomness in the choice of $G$ itself, and the generation of infection times in this $G$. Note that the definition defines error to be when approximate recovery (as defined by the sets $\B$) fails.

\begin{theorem}\label{thm:necessarycond}
In the general setting above, for {\em any} graph estimator to have a probability of error of $P_e$, we need
\begin{align*}
  m  ~ \geq ~ \frac{\left(1-P_e\right) \log \frac{|\G|}{\sup_{G'}|\mathcal{B}(G')|} - 1}{\sum_{i\in V} H(T_i)}
\end{align*}
where $H(\cdot)$ is the entropy function.
\end{theorem}
\begin{proof}
To shorten notation, we will denote $\widehat{G}(T^\U)$ simply by $\widehat{G}$. The proof uses several basic information-theoretic inequalities, which can be found e.g. in \cite{CovTho06}. In the following $H(\cdot)$ denotes entropy and $I(\cdot;\cdot)$ denotes mutual information.

We can see that the following diagram forms a Markov chain
\begin{align*}
  G \longleftrightarrow T^{\U} \longleftrightarrow \widehat{G}
\end{align*}
We have the following series of inequalities:
\begin{align*}
  H(G) 	&= I(G;\widehat{G}) + H(G \; \vert \; \widehat{G}) \\
	&\stackrel{(\varsigma_1)}{\leq} I(G;T^U) + H(G \; \vert \; \widehat{G}) \\
	&\stackrel{(\varsigma_2)}{\leq} H(T^U) + H(G \; \vert \; \widehat{G}) \\
	&\stackrel{(\varsigma_3)}{\leq} mH(T) + H(G \; \vert \; \widehat{G}) \\
	&\stackrel{(\varsigma_4)}{\leq} m\sum_{i\in V}H(T_i) + H(G \; \vert \; \widehat{G}) \\
\end{align*}
where $(\varsigma_1)$ follows from the data processing inequality, $(\varsigma_2)$ follows from the fact
that the mutual information between two random variables is less than the entropy of either of them,
$(\varsigma_3)$ and $(\varsigma_4)$ follows from the subadditivity of entropy.
Since $G$ is sampled uniformly at random from $\G$, we have that $H(G) = \log |\G|$.
We now use Fano's inequality to bound $H(G\;\vert\;\widehat{G})$.
\begin{align*}
  H(G \; \vert \; \widehat{G}) &\stackrel{(\varsigma_1)}{\leq} H(G,Err \; \vert \; \widehat{G}) \\
	&\stackrel{(\varsigma_2)}{\leq} H(Err \; \vert \; \widehat{G}) + H(G \; \vert \; Err,\widehat{G}) \\
	&\stackrel{(\varsigma_3)}{\leq} H(Err) + H(G \; \vert \; E,\widehat{G}) \\
	&\stackrel{(\varsigma_4)}{\leq} 1 + P_e \log |\G| + (1-P_e) \log \sup_{\widehat{G}}|\mathcal{B}_s(\widehat{G})|
\end{align*}
where $Err$ is the error indicator random variable (i.e., is $1$ if $G\notin \B(\widehat{G})$ and $0$ otherwise), so that
$P_e = \mathbb{E}\left[Err\right]$.
$(\varsigma_1)$ follows from the monotonicity of entropy, $(\varsigma_2)$ follows from the chain rule of entropy,
$(\varsigma_3)$ follows from the monotonicity of entropy with respect to conditioning and $(\varsigma_4)$
follows from Fano's inequality.
Combining the above two results, we obtain
\begin{align}
  & \;\;m\sum_{i\in V}H(T_i) \geq \left(1-P_e\right) \log \frac{|\G|}{\sup_{\widehat{G}}|\mathcal{B}(\widehat{G})|} - 1 \nonumber\\
  & \Rightarrow m \geq \frac{\left(1-P_e\right) \log \frac{|\G|}{\sup_{\widehat{G}}|\mathcal{B}(\widehat{G})|} - 1}{\sum_{i\in V}H(T_i)} \label{eqn:numsamples_fano_lowerbound}
\end{align}
\end{proof}

To apply this result to a particular ensemble $\mathcal{G}$ and notion of approximation $\B$, we need to find a lower bound on $|\mathcal{G}|$, and  upper bounds on $|\B(G')|$ for all $G'$ and $H(T_i)$ for all $i$. The following lemma states an upper bound on $H(T_i)$ for our independent cascade model when we have correlation decay coefficient $\alpha$. Both our corollaries assume this is the case for all graphs in their respective ensembles.

\begin{lemma}\label{lem:entropy_upperbound}
For any graph with correlation decay coefficient $\alpha$, for any node $i$, and when $p_{init} < \frac{1}{e}$, we have that
\begin{eqnarray*} 
  H(T_i) &\leq  &\frac{p_{\textrm{init}}}{1-\alpha} \left(\log \frac{1}{p_{\textrm{init}}} +
		  \left(\frac{1-\alpha}{\alpha}\right)^2\log \frac{1}{1-\alpha}\right) \\
		& &\;\;\;  - \left(1-\frac{p_\textrm{init}}{\alpha}\right)\log\left(1-\frac{p_\textrm{init}}{\alpha}\right) \\
		& =: & p_{init} \overline{H}(\alpha,p_{init})
\end{eqnarray*}
\end{lemma}
Note that the {\em edit distance} between two graphs is the number of edges present in only one of the two graphs but not the other (i.e. the number of edges in the symmetric difference of the two graphs). Our first corollary is for the case when there is no super-graph information, and we want to approximate in global edit distance.
\begin{corollary}
 Let $\G_d$ denote the set of all graphs with in-degrees bounded by $d$, and $\B_\gamma(G')$ be the set of all graphs within edit distance $\gamma$ of $G'$. Let $p_{init} < \frac{1}{e}$. Then for any algorithm to have a probability of error of $P_e$, we need
\[
m ~ > ~ \frac{(1-P_e)}{p_{init}} \frac{1-\alpha}{\overline{H}(\alpha,p_{init})} \left ( d \log \frac{n}{d} - \frac{\gamma}{n} \log \frac{n^2}{\gamma} \right ) - 1 
\]
\end{corollary}
\begin{proof}
We have that
\begin{align*}
  \log |\G_d| &= \log {n\choose d}^n = \left(1+o(1)\right)nd \log \frac{n}{d} \\
  \log |\B_{\gamma}(G')| &\leq \log {{{n}\choose {2}}\choose {\gamma}} \leq \gamma \log \frac{n^2}{\gamma}
\end{align*}
Using the above two equations along with Theorem \ref{thm:necessarycond} and Lemma \ref{lem:entropy_upperbound} gives us the result.
\end{proof}
Note that the number of times a node is infected thus needs to be $\Omega( (d-\frac{2\gamma}{n})\log n )$ (since it is of the same order as $mp_{init}$). For exact recovery, i.e. $\gamma = 0$, we see that our result on the performance of our ML algorithm -- specialized to the no prior information case $D=n$ -- is off by just a factor $d$ in terms of the number of samples required. 

The second corollary is for the case when we do have prior supergraph information. In particular, we assume that we are given sets $\mS_i$, of size $|\mS_i| = D$, for each node $i$. We consider the ensemble $\G_{D,d}$ of all in-degree-$d$ {\em subgraphs of this fixed supergraph}. Thus for each node, we need to learn the $d$ parents it has, from a given super-set of size $D$. Finally, for each node $i$ we allow $s_i$ errors; let $\B_s(G')$ be the corresponding set of all subgraphs of the given supergraph. 
\begin{corollary}
For any estimator to have a probability of error of $P_e$ in the setting above, the number of samples $m$ must be bigger than
\[
\frac{(1-P_e)}{p_{init}} \frac{1-\alpha}{\overline{H}(\alpha,p_{init})} \left ( d \log \frac{D}{d} - \frac{1}{n} \sum_i  s_i \log \frac{eD}{s_i} + \log \max(s_i,1) \right ) - 1 
\]
%For the special case $s_i = s$ for $i \in R$ and $s_i=0$ for $i \in V\setminus R$, the above result simplifies to
%\begin{align*}
%\frac{(1-P_e)(1-\alpha)\left(d\log\frac{D}{d}-\frac{rs}{n}\log \frac{De}{s} - \frac{r}{n} \log s\right)-1}{\frac{p_{\textrm{init}}}{1-\alpha} \left(\log \frac{1}{p_{\textrm{init}}} +
%		  \left(\frac{1-\alpha}{\alpha}\right)^2\log \frac{1}{1-\alpha}\right)
%		- \left(1-\frac{p_\textrm{init}}{\alpha}\right)\log\left(1-\frac{p_\textrm{init}}{\alpha}\right)}
%\end{align*}
%where $r=|R|$.
\end{corollary}
{\bf Remark:} Specializing this result to exact recovery (i.e. $s_i = 0$) removes dependence on $n$, and again shows us that the ML algorithm is within a factor $d$ of optimal for the case when we have a super-graph.
\begin{proof}
We have the following bound on the size of the ensemble:
\begin{align*}
  \log |\mathcal{G}_d| = \log {D\choose d}^n = \left(1+o(1)\right)nd \log \frac{D}{d}
\end{align*}
Similarly,
\begin{align}
  \log |\mathcal{B}_s(\widehat{G})| &\leq \log \prod_{i\in V}\left(\sum_{l=0}^{s_i} {{D}\choose{l}}\right) \nonumber\\
	&\leq \log \prod_{i\in V}\left(\max(1,s_i) {{D}\choose{s_i}}\right) \nonumber\\
	&\leq \sum_{i\in V} \log \left(\max(1,s_i) \left(\frac{De}{s_i}\right)^{s_i}\right) \nonumber\\
	&= \sum_{i\in V} \log \max(1,s_i) + \sum_{i\in V}s_i \log \frac{De}{s_i} \label{eqn:Bs-upperbound1}
\end{align}
where 
\begin{align*}
\mathcal{B}_s(\widehat{G})= \{\widetilde{G}\in \G_d : \widetilde{\V}_i \vartriangle \widehat{\V}_i \leq s_i \;\forall \;i \in V\}
\end{align*}
Note that in the second inequality we assume $s_i \leq \frac{D}{2}$ because otherwise
if $d < \frac{D}{2}$, we can choose $\widehat{\V}_i=\Phi$ and
if $d \geq \frac{D}{2}$, we can choose $\widehat{\V}_i=\mathcal{V}_i$.
Using Theorem \ref{thm:necessarycond}, \eqref{eqn:Bs-upperbound1} and Lemma \ref{lem:entropy_upperbound} gives us the first part of the result.
\end{proof}

\section{General SIR Epidemics: Markov Graphs and Causality} \label{sec:markov}

In this section we consider a much more general model for SIR epidemics/cascades on a directed graph, and establish a connection to the classic formalism of Markov Random Fields (MRFs) -- see e.g. \cite{lauritzen-book} for a formal introduction. Specifically, we show that the (undirected) Markov graph of infection times of an SIR epidemic is obtained via the {\em moralization} of the true (directed) network graph on which the cascade spreads. A moralized graph, as defined below, is obtained by adding edges between all parents of a node (i.e. ``marrying" them), and removing all directions from all edges. Graph moralization also arises in Bayesian networks, and we comment on the relationship, and the role of causality, after we present our result.

We first briefly describe our general model for SIR epidemics, then define its Markov graph, and finally present our result.

{\bf General SIR epidemics:} We now describe a general model for SIR epidemics propagating on a directed graph. Nodes can be in one of three {\em states}: {\bf 0} for susceptible, {\bf 1} for infected and active, and {\bf 2} for resistant and inactive; we restrict our attention to discrete time in this paper. Let $X_i(t)$ be the state of node $i$ at time $t$, and $X(t)$ to be the vector corresponding to the states of all nodes. We require that this process be causal, and governed by the true directed graph $G$, in the sense that for any time step $t$,
\begin{equation}
\mathbb{P} [ X(t) = x(t) \, | \, x(0:t-1)] ~ = ~ \prod_i ~  \mathbb{P} [ X_i(t) = x_i(t) \, | \, x_{\V_i}(0:t-1)] \label{eq:markov}
\end{equation}
where the notation $x(1:t) = \{x(s),1\leq s\leq t\}$ is the entire history upto time $t$, and as before $\V_i$ is the set of parents of node $i$, and includes $i\in \V_i$ as well. Note the above encodes that the probability distribution of each node's next state depends only on the history of itself and its neighbors, but is otherwise independent of the history or current state of the other nodes. We assume that the cascade is initially seeded arbitrarily, i.e. $x(0)$ can be any fixed initial condition.

For each node $i$, let $T^{(1)}_i$ be the (random) time when its state transitioned from {\bf 0} to {\bf 1}, and $T^{(2)}_i$ for the time from {\bf 1} to {\bf 2} (of course, if neither happened then we can take them to be $\infty$). Let $T_i = (T^{(1)}_i,T^{(2)}_i)$ be the summary for node $i$'s participation in the cascade. 

{\bf Markov Graphs:} Markov random fields (MRFs, also known as Graphical Models) are a classic formalism, enabling the use of graph algorithms for tasks in statistics, physics and machine learning. The central notion therein is that of the Markov graph of a probability distribution; in particular, every collection of random variables has an associated graph. Every variable is a node in the graph, and the edges encode conditional independence: conditioned on the neighbors, the variable is independent of all the other variables. For our purposes here, the random variables are the $T := \{T_i, i\in V\}$. We say that an undirected graph $G'$ is the Markov graph of the variables $T$ if their joint probability distribution, for all $t$, factors as follows
\begin{equation*}
\mathbb{P} [ T = t ] ~ = ~ \prod_{c \in \mathcal{C}'} \, f_c(t_c) 
\end{equation*}
for some functions $f_c$; here $\mathcal{C}'$ is the set of cliques of $G'$, and for a clique $c\in \mathcal{C}'$, $t_c := \{t_i, i\in c\}$ is the vector of node times for nodes in $c$.
 
We need one more definition before we state our result.

{\bf Moralization:} Given a directed graph $G$, its moralized graph $\overline{G}$  is the undirected graph  where two nodes are connected if and only if they either have a parent-child relationship in $G$, or if they have a common child, or both. Formally, undirected edge $(i,j)$ is present in $\overline{G}$ if and only if at least one of the following is true \\
{\em (a)} directed edges $(i,j)$ or $(j,i)$ are present in $G$, or \\
{\em (b)} there is some node $k$ such that $(i,k)$ and $(j,k)$ are present in $G$ (i.e. $k$ is a common child).\\
Figure \ref{fig:moralization-example} illustrates the process of moralization with an example.
\begin{figure}[h]
\centering
\subfigure[Directed graph $G$]{
\includegraphics[width=0.35\linewidth]{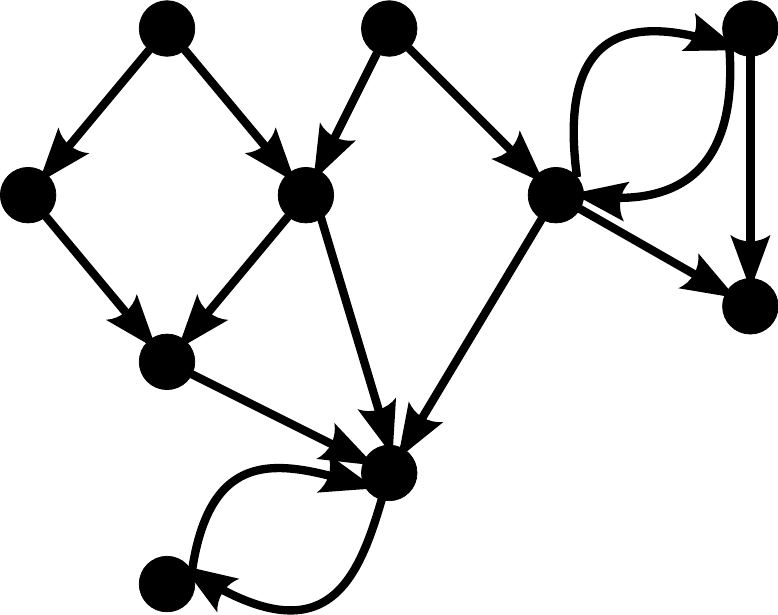}
\label{fig:moralization-dirgraph}
}
\subfigure[Moralized graph $\overline{G}$ of $G$]{
\includegraphics[width=0.35\linewidth]{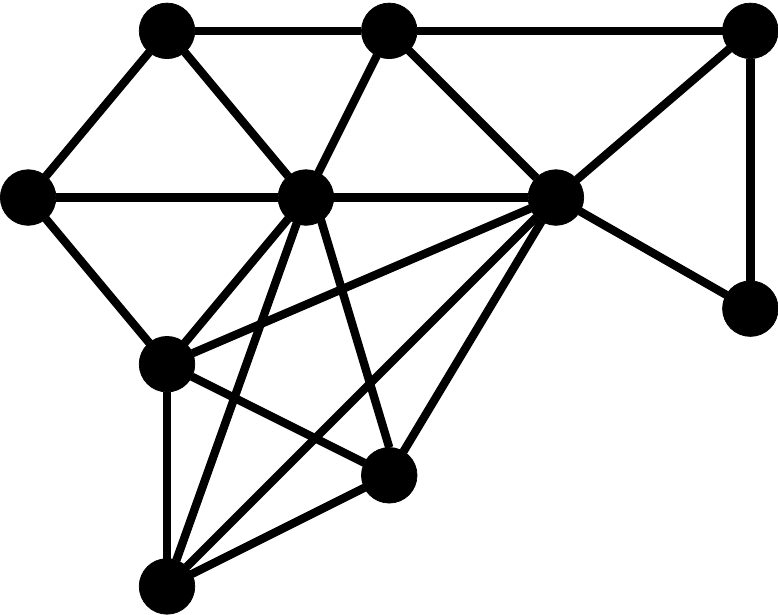}
\label{fig:moralization-undirgraph}
}
\caption{An example of moralization}
\label{fig:moralization-example}
\end{figure}

\begin{theorem}\label{thm:markov}
Suppose infection times $T$ are generated from a general SIR epidemic, as above, propagating on a directed network graph $G$. Let $\overline{G}$ be the (undirected) moralized graph of $G$. Then, $\overline{G}$ is the Markov graph of $T$.
\end{theorem}

{\bf Remarks:} The main appeal of this result arises from the generality of the model; indeed, it may be possible to learn the moralized graph even when we may not know what the precise epidemic evolution model is, as long as it satisfies (\ref{eq:markov}). In particular, related to the focus of this paper, there has been substantial work on learning the Markov graph structure of random variables from samples. In our setting, each cascade is a sample from the joint distribution of $T$, and hence one can imagine using some of these techniques. Markov graph learning techniques can generally be divided into \\
{\em (a)} those that assume a specific class of probability distributions: see e.g. \cite{MeiBuh06, Ravikumaretal08} for Gaussian MRFs, \cite{RavWaiLaf10, AnaTan11} for Ising models, \cite{Jalalietal11} for general discrete pairwise distributions. These typically require knowledge of the precise parametric form of the dependence, but then enable learning with a smaller number of samples. \\
{\em (b)} distribution-free algorithms, usually for discrete distributions and based on conditional independence tests \cite{AbbKolNg06,BreMosSly08,NetBanSanSha10}. These do not need to know the parametric form a-priori, but typically have higher computational and sample complexity.

{\bf Causality:} It is interesting to contrast Theorem \ref{thm:markov} with the other results in this paper. In particular, on the one hand, Theorems \ref{thm:samplecomplexity_maxLL} and \ref{thm:greedy-guarantees} utilize the fact precise causal process that generates $T$ to find the exact true  directed network graph. On the other hand, applying a Markov graph learning technique directly to the samples of $T$, without leveraging the process that generated them, only allows us to get to the moralized graph. It thus serves as a motivating example to extend the study of graph learning from samples to causal phenomena, in a way that explicitly takes into account time dynamics. 

Moralization also arises in Bayesian networks; this is an alternative formulation that associates an {\em acyclic} directed graph with a probability distribution. In that setting, the undirected Markov graph is also the moralization of this directed graph. We note however that our original true network graph $G$ can have directed cycles; in our setting the moralization arises from (ignoring the) causality in time.
\section{Experiments}\label{sec:experiments}

As an initial empirical illustration of our results, in this section, we present -- via Figures \ref{fig:sample_complexity}, \ref{fig:super_graph}, \ref{fig:twitter-subsubgraph-recovery} and \ref{fig:twitter-300node-scatterplot} --  empirical evaluations of both the ML and Greedy algorithms on synthetic graphs, and sub-graphs of the Twitter graph. In all cases, for the ML algorithm the threshold $\eta$ was picked via cross-validation. 

\begin{figure}[h]
\centering
\subfigure{
\includegraphics[width=0.47\linewidth , height = 0.33\linewidth]{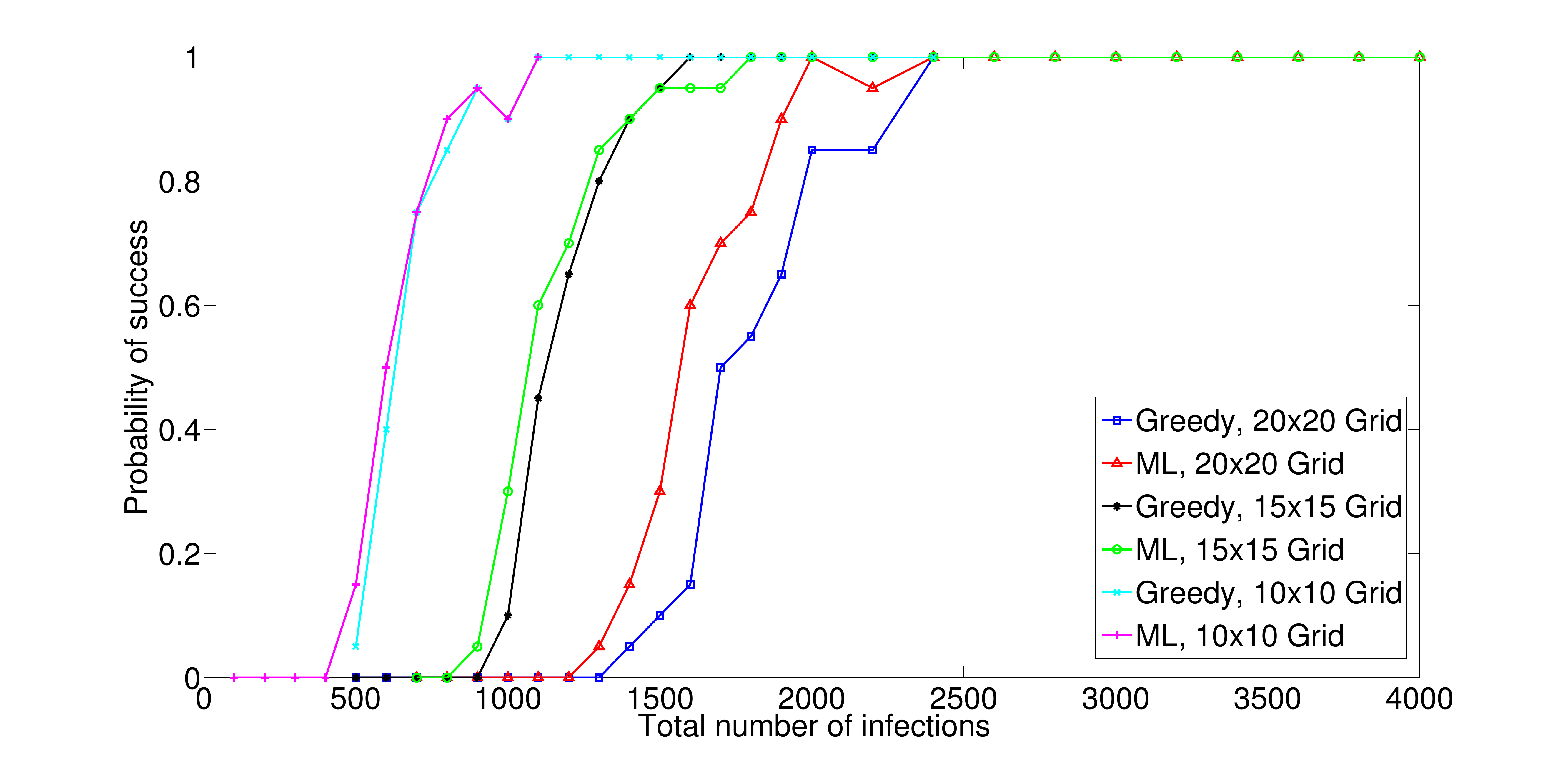}
%\label{fig:moralization-dirgraph}
}
\subfigure{
\includegraphics[width=0.47\linewidth , height = 0.33\linewidth]{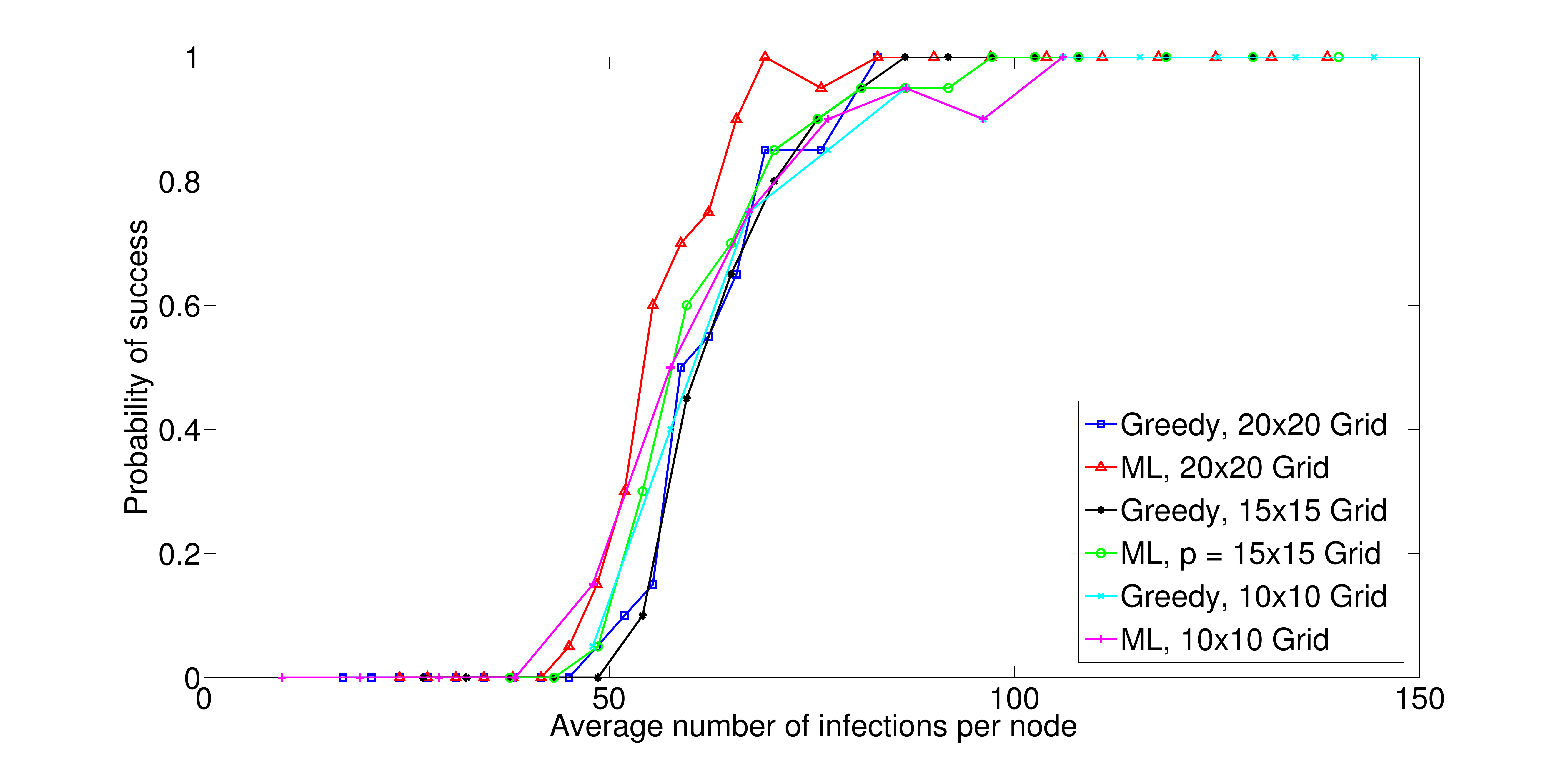}
%\label{fig:moralization-undirgraph}
}
\caption{{\bf Interpreting sample complexity:} As mentioned in Section \ref{sec:sysmodel}, and re-inforced by our theorems, consistent structure recovery is governed not so much by the total number of infections $m$ in the network, as by the number of times a node is infected (which is approx. $p_{init}m$). This figure provides some empirical validation of this claim; the plots on the left and right are from the same set of experiments using the ML algorithm (and no super-graph information). On the left we plot the probability of successful recovery of the entire graph as a function of $m$, while on the right we plot it as a function of the average number of times a node was infected; for several different sizes of 2-d grids. On the left, we see that the total number of cascades varies noticeably with grid size, but the average number of infections does not. This squares with Theorem \ref{thm:samplecomplexity_maxLL}, since in all these graphs the $d$ is the same, and $\log n$ does not vary much either. }
\label{fig:sample_complexity}
\end{figure}

\begin{figure}[h]
\centering
\includegraphics[width=0.8\linewidth]{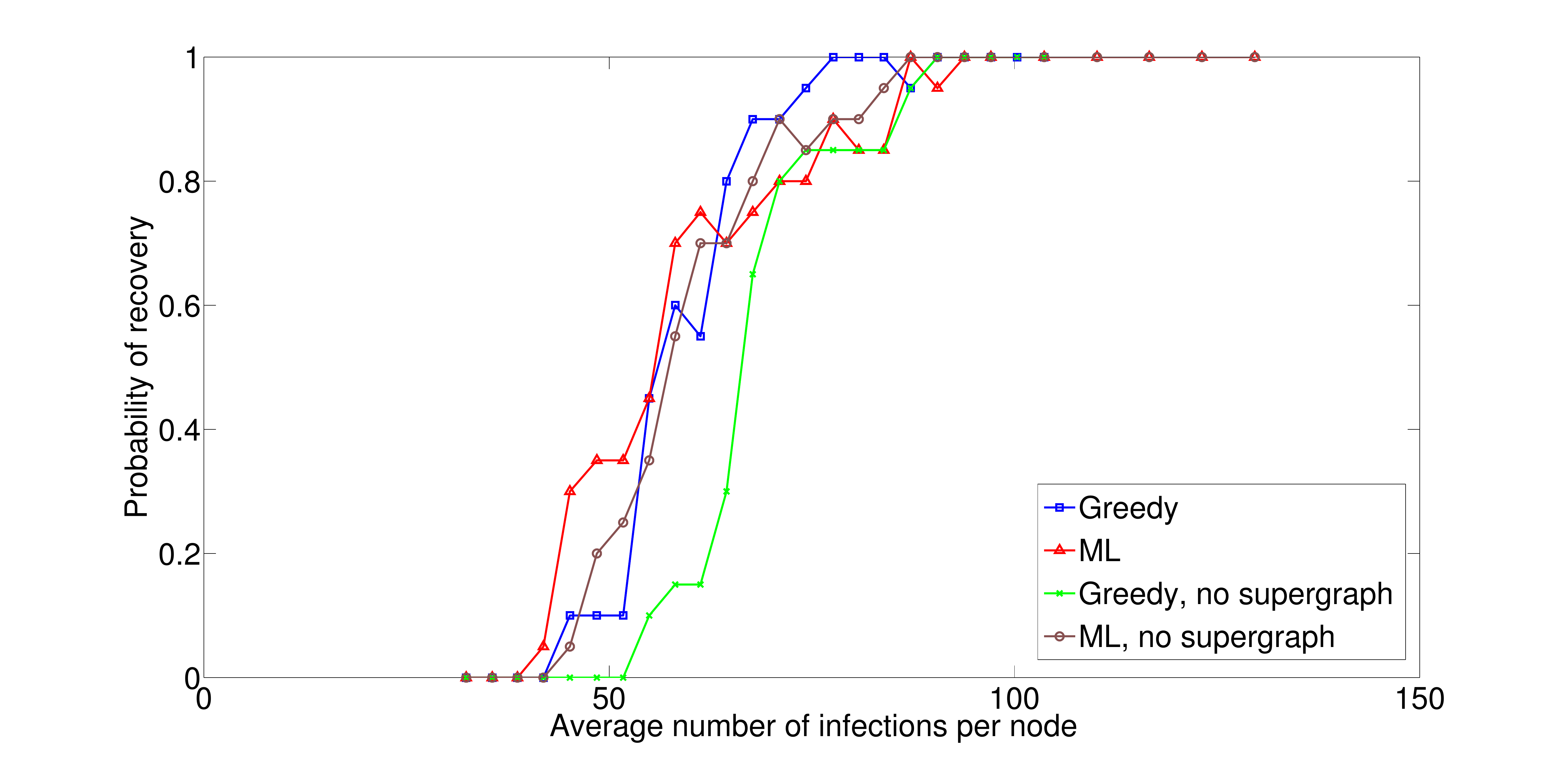}
\caption{{\bf Effect of super-graph information:} The presence of super-graph information can reduce the number of node infections (and hence cascades) required to learn the graph. Here we plot the probability of successful recovery for a 200-node random 4-regular graph, for the ML and Greedy algorithms, for two scenarios: when we are given a super-graph of regular degree 8 that contains the true graph, and when we are not given such information. We can see that the extent of reduction in sample complexity is moderate, reflecting the fact that the effect of super-graph information is logarithmic (i.e. $\log D$ vs $\log n$). 
}
\label{fig:super_graph}
\end{figure}
\begin{figure}[h]
\centering
\includegraphics[width=1\linewidth]{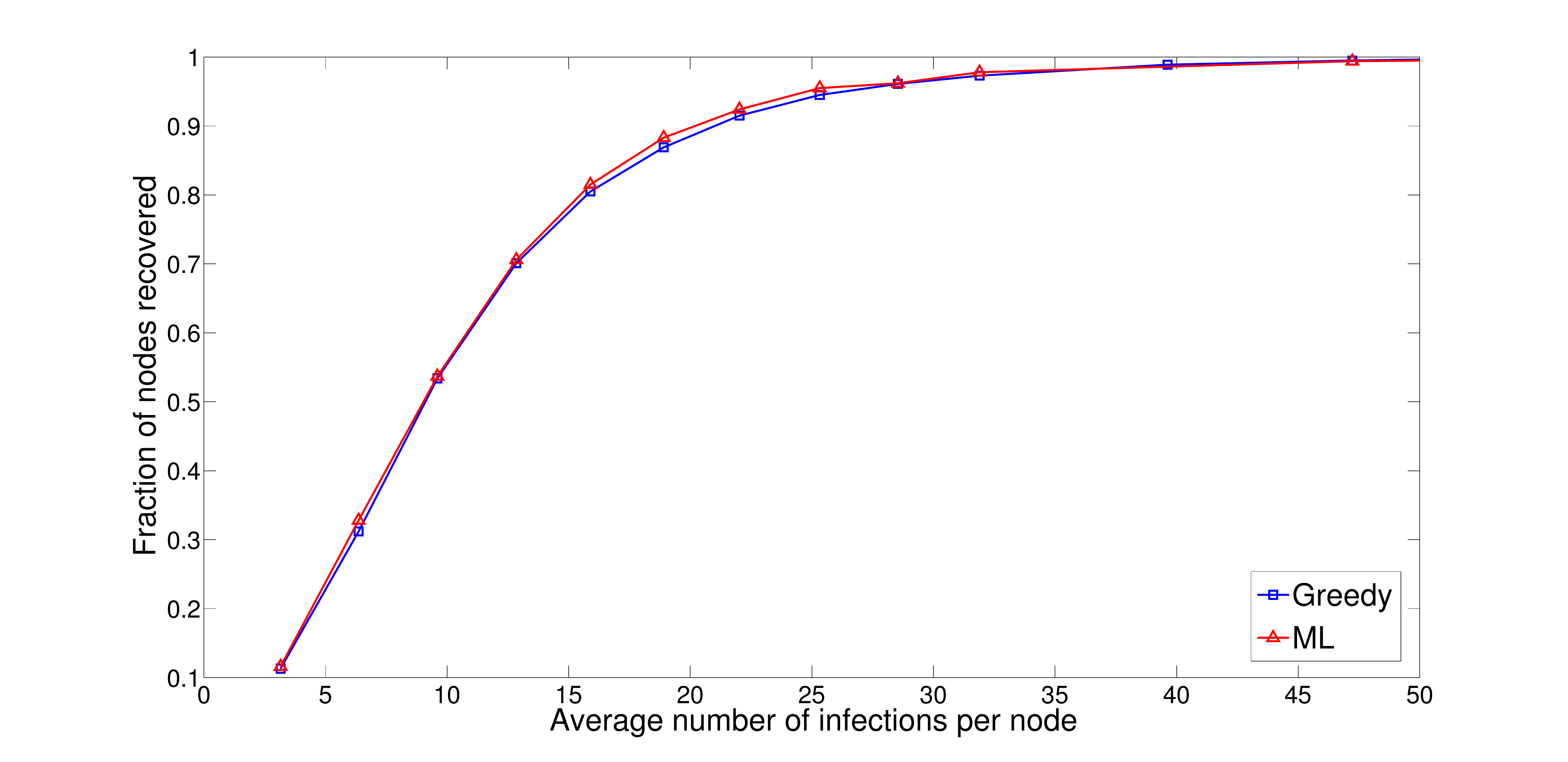}
\caption{{\bf Twitter graph as a super-graph:} We extracted a directed graph of $1000$ nodes from Twitter as follows: if $i$ follows $j$ on Twitter (i.e. $j$'s tweets appear in $i$'s feed) then we put a directed edge $j \rightarrow i$ in the graph. We then treat this as a super-graph; from this we chose a 4-regular sub-graph of ``real" neighbors -- i.e. each person is assumed to have at most 4 significant influencers from among  the people she/he follows. Edge parameters were chosen so as to give weight to only the chosen important parents. Infections were
sampled by simulating the independent cascade model using the above parameters. Both ML and Greedy algorithms were
given the infections as input along with the $1000$ node graph as a super graph. For various number of total infections,
the x-axis shows the average number of infections of a node and the y-axis shows the fraction of nodes whose neighborhoods
are recovered exactly by the algorithm.
}
\label{fig:twitter-subsubgraph-recovery}
\end{figure}

%In the second experiment, a $300$ node subgraph of the Twitter graph was extracted. As was the case in the $1000$ node
%sub graph, even in the $300$ node subgraph, there were nodes of high in-degree and/or out-degree. For each node,
%all parents were assigned equal edge parameters. Infections were sampled according to the independent cascade model
%using these edge parameters. Both the ML and Greedy algorithms were given these infections as input without any
%super graph information. Figure \ref{fig:twitter-300node-scatterplot} shows a scatter plot of the number of infections
%taken by a node for its neighborhood to be estimated correctly versus the degree of that node.
\begin{figure}[h]
\centering
\includegraphics[width=1\linewidth]{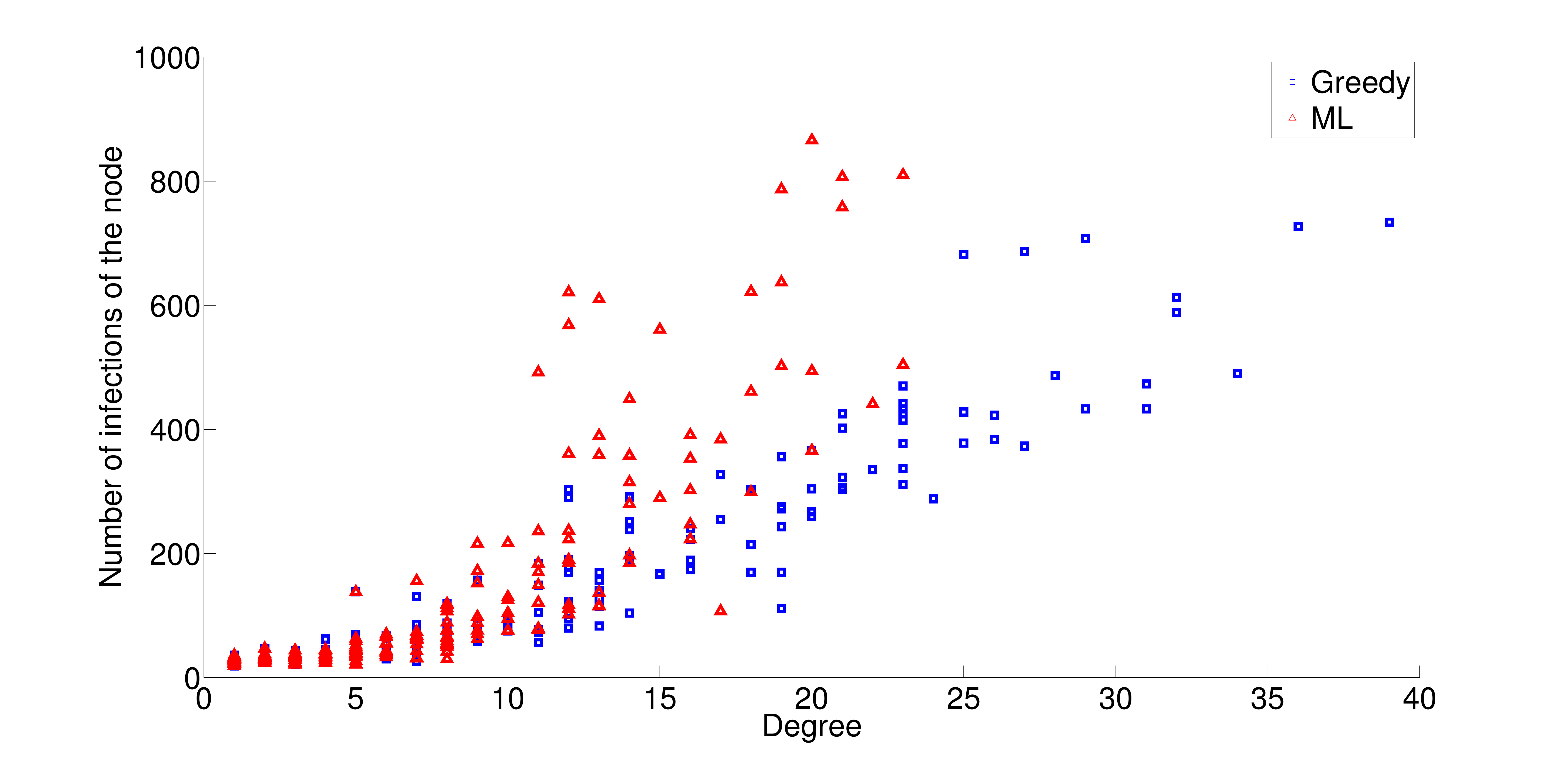}
\caption{{\bf Dependence on degree for Twitter graph:} This figure is a scatter plot where for each node we plot its degree, and the number of times it was infected before ML or Greedy succeeded in finding its neighborhood. The graph is a  $300$ node graph was extracted from Twitter (with edges made as explained in Figure \ref{fig:twitter-subsubgraph-recovery}); now however this is treated as the true graph to be learnt, and the algorithm is given no super-graph information. At least for this example, the sample complexity increases super-linearly with degree for the ML algorithm, where as
the dependence of the sample complexity on the degree is almost linear for the Greedy algorithm. This is in spite of the fact that the graph is far from being a tree; it has several small cycles.
}
\label{fig:twitter-300node-scatterplot}
\end{figure}

\section{Summary and Discussion}

This paper studies the problem of learning the graph on which epidemic cascades spread, given only the times when nodes get infected, and possibly a super-graph. We studied the sample complexity -- i.e. the number of cascade samples required -- for two natural algorithms for graph recovery, and also established a corresponding information-theoretic lower bound. To our knowledge, this is the first paper to study the sample complexity of learning graphs of epidemic cascades. Several extensions suggest themselves; we discuss some below.

{\bf Observation Model:} In this paper it is assumed that we have access to the times when nodes get infected. However, this may not always be possible. Indeed a weaker assumption is to only know the infected set in each cascade. To us this seems like a much harder problem, e.g. it is now not clear that there is a decoupling of the global graph learning problem.

{\bf Decoupling:} A key step in our ML results is to show that the global graph finding problem decouples into $n$ local problems. Our proof of this fact can be extended to any causal network process -- i.e. any process where the state $x_i(t)$ depends only on $x_{\V_i}(t-1)$ -- under the assumption that we can reconstruct the entire process trajectory from our observations (so e.g. the weaker observation model above would not fall into this class). In particular, it holds for more general models of epidemic cascade propagation as well; we focused on the discrete-time one-step model as a first step. 

{\bf Correlation decay:} Our results are for the case of correlation decay, i.e. when the cascade from one seed reaches a constant depth of nodes before extinguishing. Equally interesting and relevant is the case without correlation decay, when the cascade from each seed can reach as much as a constant fraction of the network. We suspect, based on experiments, that our algorithms would be efficient in this case as well; however, a proof would be technically quite different, and interesting.

{\bf Greedy algorithms:} As can be seen in our experiments, the greedy algorithm performs quite well even when the graph is far from being a tree (i.e. has several small cycles). It would be interesting to develop an alternate and more general proof of the performance of the greedy algorithm. We also note that one can easily formulate greedy algorithms in more general epidemic settings; this would involve iteratively choosing the parameter that gives the biggest change in the corresponding likelihood function. 

% \bibliographystyle{abbrv}
% \bibliography{infectionrefs} 

\begin{thebibliography}{10}

\bibitem{AbbKolNg06}
P.~Abbeel, D.~Koller, and A.~Y. Ng.
\newblock Learning factor graphs in polynomial time and sample complexity.
\newblock {\em Journal of Machine Learning Research}, 7:1743--1788, 2006.

\bibitem{AnaTan11}
A.~Anandkumar and V.~Y.~F. Tan.
\newblock {High-Dimensional Structure Learning of Ising Models : Tractable
  Graph Families}.
\newblock {\em Preprint}, June 2011.

\bibitem{Bonetal08}
T.~Bonald, L.~Massouli\'{e}, F.~Mathieu, D.~Perino, and A.~Twigg.
\newblock Epidemic live streaming: optimal performance trade-offs.
\newblock {\em SIGMETRICS Perform. Eval. Rev.}, 36:325--336, June 2008.

\bibitem{BreMosSly08}
G.~Bresler, E.~Mossel, and A.~Sly.
\newblock Reconstruction of markov random fields from samples: Some
  observations and algorithms.
\newblock In {\em APPROX '08 / RANDOM '08}, pages 343--356, Berlin, Heidelberg,
  2008. Springer-Verlag.

\bibitem{CovTho06}
T.~M. Cover and J.~A. Thomas.
\newblock {\em Elements of Information Theory (Wiley Series in
  Telecommunications and Signal Processing)}.
\newblock Wiley-Interscience, 2006.

\bibitem{GolLibMul01}
J.~Goldenberg, B.~Libai, and E.~Muller.
\newblock Talk of the network: A complex systems look at the underlying process
  of word-of-mouth.
\newblock {\em Marketing Letters}, 12:211--223, 2001.

\bibitem{Gruhletal04}
D.~Gruhl, R.~Guha, D.~Liben-Nowell, and A.~Tomkins.
\newblock Information diffusion through blogspace.
\newblock In {\em Proc. 13th International Conference on World Wide Web}, WWW
  '04, pages 491--501, New York, NY, USA, 2004. ACM.

\bibitem{Jalalietal11}
A.~Jalali, P.~Ravikumar, V.~Vasuki, and S.~Sanghavi.
\newblock On learning discrete graphical models using group-sparse
  regularization.
\newblock In {\em Proceedings of the Fourteenth International Conference on
  Artificial Intelligence and Statistics}, volume~15, pages 378--387, 2011.

\bibitem{KemKleTar03}
D.~Kempe, J.~Kleinberg, and E.~Tardos.
\newblock Maximizing the spread of influence through a social network.
\newblock In {\em Proc. 9th ACM SIGKDD international conference on Knowledge
  discovery and data mining}, KDD '03, pages 137--146, New York, NY, USA, 2003.
  ACM.

\bibitem{KepWhi91}
J.~O. Kephart and S.~R. White.
\newblock Directed-graph epidemiological models of computer viruses.
\newblock {\em Security and Privacy, IEEE Symposium on}, 0:343, 1991.

\bibitem{Kosetal99}
D.~Kosterev, C.~Taylor, and W.~Mittelstadt.
\newblock Model validation for the august 10, 1996 wscc system outage.
\newblock {\em Power Systems, IEEE Transactions on}, 14(3):967 --979, aug 1999.

\bibitem{lauritzen-book}
S.~Lauritzen.
\newblock {\em Graphical models}.
\newblock Oxford University Press, 1996.

\bibitem{LerGho10}
K.~Lerman and R.~Ghosh.
\newblock Information contagion: An empirical study of the spread of news on
  {D}igg and {T}witter social networks.
\newblock In {\em Proc. International AAAI Conference on Weblogs and Social
  Media}, 2010.

\bibitem{MasTwi08}
L.~Massouli\'{e} and A.~Twigg.
\newblock Rate-optimal schemes for peer-to-peer live streaming.
\newblock {\em Performance Evaluation}, 65(11-12):804 -- 822, 2008.

\bibitem{MeiBuh06}
N.~Meinshausen and P.~B\"{u}hlmann.
\newblock High-dimensional graphs and variable selection with the lasso.
\newblock {\em Annals of Statistics}, 34:1436--1462, 2006.

\bibitem{MyeLes10}
S.~A. Myers and J.~Leskovec.
\newblock On the convexity of latent social network inference.
\newblock In {\em Proc. Neural Information Processing Systems (NIPS)}, 2010.

\bibitem{NetBanSanSha10}
P.~Netrapalli, S.~Banerjee, S.~Sanghavi, and S.~Shakkottai.
\newblock Greedy learning of markov network structure.
\newblock In {\em Communication, Control, and Computing (Allerton), 2010 48th
  Annual Allerton Conference on}, pages 1295 --1302, sept 29 - oct. 1 2010.

\bibitem{Per99}
C.~Perrow.
\newblock {\em {Normal Accidents: Living with High-Risk Technologies}}.
\newblock Princeton University Press, updated edition, Sept. 1999.

\bibitem{Poor_book}
V.~Poor.
\newblock {\em An Introduction to Signal Detection and Estimation}.
\newblock Springer, 1994.

\bibitem{RavWaiLaf10}
P.~Ravikumar, M.~J. Wainwright, and J.~D. Lafferty.
\newblock High-dimensional graphical model selection using $l_1$-regularized
  logistic regression.
\newblock {\em Annals of Statistics}, 38(3):1287--1319, 2010.

\bibitem{Ravikumaretal08}
P.~Ravikumar, M.~J. Wainwright, G.~Raskutti, and B.~Yu.
\newblock {Model selection in Gaussian graphical models: High-dimensional
  consistency of l1-regularized MLE.}
\newblock 2008.

\bibitem{RodLesKra10}
M.~G. Rodriguez, J.~Leskovec, and A.~Krause.
\newblock Inferring networks of diffusion and influence.
\newblock In {\em Proc. 16th ACM SIGKDD international conference on Knowledge
  discovery and data mining}, KDD '10, pages 1019--1028, New York, NY, USA,
  2010. ACM.

\bibitem{Sacetal00}
M.~L. Sachtjen, B.~A. Carreras, and V.~E. Lynch.
\newblock Disturbances in a power transmission system.
\newblock {\em Phys. Rev. E}, 61:4877--4882, May 2000.

\bibitem{Zhou10}
Z.~Zhou, R.~Bandari, J.~Kong, H.~Qian, and V.~Roychowdhury.
\newblock Information resonance on {T}witter: watching {I}ran.
\newblock In {\em Proc. 1st Workshop on Social Media Analytics}, SOMA '10,
  pages 123--131, New York, NY, USA, 2010. ACM.

\end{thebibliography}

\appendix
\section{Correlation decay}

\begin{proof}[Proof of Lemma \ref{lem:corrdecay-probbound}]
We establish this by an induction on the number of nodes $n$ in the graph. If $n=1$, the statement above is obvious. Suppose the statement above is true for all graphs which have upto $n-1$ nodes. Consider now a graph
$G$ that has $n$ nodes. Consider any node $i$. The statement of the proposition is clearly true for $t=1$. For $t>1$, consider the probability that $i$
is infected by a parent $k\in \mS_i$ at time step $t$. This can be upper bounded as follows:
\begin{eqnarray*}
  \mathbb{P}_G\left[\text{$k$ infects $i$ at time $t$}\right] &\leq &\mathbb{P}_{\widetilde{G}}\left[T_k=t-1\right] p_{ki} \\
			&\leq & \left(1-\alpha\right)^{t-2} p_{\textrm{init}}p_{ki}
\end{eqnarray*}
where $\widetilde{G} := G\setminus i$ is the graph without node $i$, $\mathbb{P}_G$ denotes the probability when the graph is $G$, and similarly for $\mathbb{P}_{\widetilde{G}}$. The second inequality follows from the induction assumption, and the fact that if $\alpha$ is the decay coefficient for $G$, it is also for $\widetilde{G}$.
 Taking a union bound over $k\in \mS_i$ now gives us the statement of the theorem for $G$:
\begin{eqnarray*}
  \mathbb{P}_G \left[T_i = t \right] &= & \sum_{k \in \mS_i} \mathbb{P}_G\left[\text{$k$ infects $i$ at time $t$}\right]  \\
		&\leq & \left(1-\alpha\right)^{t-2} p_{\textrm{init}} \sum_{k\in \mS_i} p_{ki} \\
		&\leq & \left(1-\alpha\right)^{t-1} p_{\textrm{init}}
\end{eqnarray*}
The bounds on $\mathbb{P}[T_i < \infty]$ follow simply from summing this geometric series.
\end{proof}

%%%%%%%%%%%%%%%%%%%%%%

\section{Maximum Likelihood}

\vspace{0.05in}
\subsection{Proof of Prop. \ref{prop:likelihood}}

Let $X_i(\tau)=0$ if $i$ is susceptible at time $\tau$, $1$ if $i$ is active at time $\tau$ and $2$ if $i$ is inactive
at time $\tau$. Let $X(\tau)$, $\tau = 0,\cdots,n$ be the corresponding vector process. Note that $X(\tau)$ is a Markov process, and
there is a one to one correspondence between the set of infection times $t$ and sample path $x(\tau)$ of the process $X(\tau)$.

Given $t$, let $x^0(\tau)$ be the corresponding vector process. In particular,
\begin{align*}
x_i^0(\tau) = 
  \begin{cases}
    0 & \mbox{ if } \tau < t_i \\
    1 & \mbox{ if } \tau = t_i \\
    2 & \mbox{ if } \tau > t_i
  \end{cases}
\end{align*}
Then,
\begin{align*}
 \mathbb{P}_{\theta}\left[T=t\right] & = ~ \mathbb{P}_{\theta}\left[X(\tau) = x^0(\tau) \mbox{ for } \tau = 0,\cdots,n\right] \\
	&= \mathbb{P}_{\theta}\left[X(0) = x^0(0)\right] ~ \times \\
	&\;\;\;\;\;  \prod_{\tau = 1}^n \mathbb{P}_{\theta}\left[X(\tau) = x^0(\tau) \middle \vert X(\tau-1) = x^0(\tau-1)\right]
\end{align*}
Now, $\mathbb{P}_{\theta}\left[X(0) = x^0(0)\right] = p_{\textrm{init}}^s\left(1-p_{\textrm{init}}\right)^{n-s}$.
Also,
\begin{align*}
  &\mathbb{P}_{\theta}\left[X(\tau) = x^0(\tau) \middle \vert X(\tau-1) = x^0(\tau-1)\right] \\
  &\quad = ~ \prod_{i\in V} \mathbb{P}_{\theta}\left[X_i(\tau) = x_i^0(\tau) \middle \vert X(\tau-1) = x^0(\tau-1)\right]
\end{align*}
because each node gets infected independently from each of its currently active neighbors. Thus we have that
\begin{align}\label{eqn:maxLL-probproduct}
  \mathbb{P}\left[T=t\right]
      = p_{\textrm{init}}^s\left(1-p_{\textrm{init}}\right)^{n-s} 
	    \prod_{i\in V} \left(\prod_{\tau=1}^n a_i(\tau)\right)
\end{align}
where $a_i(\tau) = \mathbb{P}_{\theta}\left[X_i(\tau) = x_i^0(\tau) \middle \vert X(\tau-1) = x^0(\tau-1)\right]$.
It is clear that for $\tau > t_i$, $a_i(\tau)=1$. For $\tau=t_i$, $a_i(\tau)$ is the probability that at least one
of its active nodes at time $t_i-1$ infected node $i$. Thus, 
\begin{align}\label{eqn:maxLL-infectionpart}
  a_i(t_i) = 1 - \prod_{j:t_j = t_i-1} \exp\left(-\theta_{ji}\right)
\end{align}
Finally, for each $\tau < t_i$, $a_i(\tau)$ is the probability that active nodes at time $\tau-1$ failed to infect node $i$.
The set of all nodes that were active but failed to infect susceptible node $i$ is $\{j:t_j \leq t_i-2\}$. So we have
\begin{align}\label{eqn:maxLL-failedinfectionpart}
  \prod_{\tau<t_i} a_i(\tau) = \prod_{j:t_j \leq t_i-2} \exp\left(-\theta_{ji}\right)
\end{align}
Putting \eqref{eqn:maxLL-probproduct}, \eqref{eqn:maxLL-infectionpart} and \eqref{eqn:maxLL-failedinfectionpart} together
and taking log gives the result. 

Concavity follows from the fact that $\log(1 - \exp(-x))$ is a concave function of $x$, and the fact that if any function $f(x)$ is a concave function of $x$ then $f(\sum_i \theta_i)$ is jointly concave in $\theta$.
$\blacksquare$

\subsection{Proof of Theorem \ref{thm:samplecomplexity_maxLL}}

We focus on the recovery of the neighborhood of node $i$. For brevity, we will drop $i$ from sub-scripts; thus we denote $\theta_{*i}$ by $\theta$, $\V_i$ by $\V$ and $\mS_i$ by $\mS$, and $d_i,D_i$ by $d,D$. Let $\theta^*$ be the true parameter values. Define the empirical log-likelihood function by
\begin{align*}
 \widehat{L}(\theta) \defas \frac{1}{m} \sum_u \mL_i(t^u;\theta)
\end{align*}
Note that the ML algorithm finds $\widehat{\theta}=\argmax_{\theta} \widehat{L}(\theta)$.
Also let $L(\theta)\defas \mathbb{E}_{\theta^*}\left[\mL_i(T,\theta)\right]$.

{\bf Idea:} Note that as the number of samples $m$ increases, $\widehat{L}\rightarrow L$. Also, we know that
$\theta^*=\arg \min_{\theta} L(\theta)$; this is just stating that the expected value of
the likelihood function is maximized by the true parameter values, a simple classical
result from ML estimation \cite{Poor_book}. Thus when $\widehat{L}\simeq L$, their minimizers will also be close; i.e. 
$\theta^* \simeq \widehat{\theta}$. However, they will not be exactly equal; hence hope then is to have subsequent thresholding find the significant edges. The challenge is in establishing non-asymptotic bounds that show that $m$ scales much slower than $n$ (the network size) or $D$ (the size of the super-neighborhood).

\noindent {\bf Roadmap to the proof:} 

{\bf (a)} In Proposition \ref{prop:grad-characterization} we provide an expression for the gradient $\grad_j L(\theta^*)$ of the expected log-likelihood evaluated at the true parameters $\theta^*$. This can be used to show that $\grad_j L(\theta^*) = 0$ for the true neighbors $j\in \V$, and  for the others we can show that $\grad_j L(\theta^*) < 0$ for $j\notin \V$. 

{\bf (b)} Note that {\em if} we had similar relationships hold for the empirical likelihood, i.e. if $\grad_j \widehat{L}(\widehat{\theta}) = 0$ for $j\in \V$ and $\grad_j \widehat{L}(\widehat{\theta}) < 0$ for $j\notin \V$, then we would be done; this is because by complementary slackness conditions we would have that $\widehat{\theta}_j > 0$ for $j\in \V$ and $\widehat{\theta}_j = 0$ otherwise: the non-zero $\widehat{\theta}_j$ would then correspond to the true neighborhood. Of course, these relationships do not hold exactly; the rest of the proof is showing they hold approximately, and the neighborhood can be found by thresholding. 

{\bf (c)} As a first step to analyzing $\grad_j \widehat{L}(\widehat{\theta})$, in Lemma \ref{lemma:concentration-results} we establish concentration results showing that an intermediate quantity $\grad_j \widehat{L}(\theta^*)$ is close to $\grad_j L(\theta^*)$, and hence we can show that $|\grad_j \widehat{L}(\theta^*)|< a$ for $j\in \V$ (i.e. the gradient is small for the true neighbors), and $\grad_j \widehat{L}(\theta^*)<-b$ for $j\notin \V$ (i.e. the gradient is negative for the others). Here $a$ and $b$ depend on the system parameters, and $a$ depends on the threshold $\eta$ as well, with $a \rightarrow 0$ as $\eta \rightarrow 0$. This latter dependence is important as it shows that once the number of samples $m$ becomes large, we can choose $\eta$ small and get exact recovery. 

{\bf (d)} In Lemma \ref{lemma:thetaneighb-upperbound}, we provide an upper bound on the value of $\widehat{\theta}_j$ for $j\in \V$. We need this to not be too large for the next step. 

{\bf (e)} In Lemma \ref{lemma:thetanonneighb-upperbound} we derive an upper bound on the total value $\sum_{j\notin \V} \widehat{\theta}_j$ of the non-neighbor parameters in $\widehat{\theta}$. This upper bound implies that no non-neighbors will be selected after thresholding at $\eta$, completing the proof of the first claim of the theorem. 

{\bf (f)} Finally, in Lemma \ref{lemma:thetaneighb-lowerbound} we show that, for true neighbors $j\in \V$, if the true $p^*_{ji} > \frac{8}{\alpha} (e^{2\eta} - 1)$ then $\widehat{\theta}_j > \eta$, and will thus be estimated to be in the true neighborhood. This completes the proof of the second claim of the theorem.

%%%%%%%%%%%%%%%%%%%%%%%

\begin{prop}\label{prop:grad-characterization}
\begin{align}
  \grad_j L(\theta^*) = -\mathbb{P}\left[T_i > T_j~ ; ~ T_k \neq T_j ~~ \forall ~ k \in \V\right]
  \label{eqn:gradient}
\end{align}
\end{prop}

\begin{proof}
Taking the derivative of $L(\cdot)$ with respect to $\theta_j$, we obtain
\begin{align*}
  \grad_j L(\theta)
  = \mathbb{E}\left[ -\mathds{1}_{\{T_j \leq T_i - 2\}} +
		    \frac{\mathds{1}_{\{T_j = T_i - 1\}}}
		      {\exp\left(\sum_{k\colon T_k = T_i - 1} \theta_{k}\right)-1}\right]
\end{align*}
Let $\mathcal{F}_{T_j}$ be the sigma algebra with information up to the (random) time $T_j$.
By iterated conditioning, we obtain
\begin{align}
  \grad_j L(\theta^*)
		      = -\mathbb{E}\left[ \mathbb{E}\left[ \mathds{1}_{\{T_j \leq T_i - 2\}} -
		      \frac{\mathds{1}_{\{T_j = T_i - 1\}}}
		      {\exp\left(\sum_{k\colon T_k = T_i - 1} \theta^*_{ki}\right)-1}
		      \;\middle\vert\; \mathcal{F}_{T_j} \right] \right] \label{eqn:grad-conditional}
\end{align}
Since the event $\{T_i \leq T_j\}$ is measurable in $\mathcal{F}_{T_j}$, we have
\begin{align}
&\mathbb{E}\left[ \mathds{1}_{\{T_j \leq T_i - 2\}} -
     \frac{\mathds{1}_{\{T_j = T_i - 1\}}}
     {\exp\left(\sum_{k\colon T_k = T_i - 1} \theta^*_{ki}\right)-1}
     \;\middle\vert\; \mathcal{F}_{T_j} \right]
	= 0 \mbox{ if } T_i \leq T_j \label{eqn:grad-conditional-1}
\end{align}
On the other hand, if $\{T_i > T_j\}$, we have
\begin{align*}
&\mathbb{E}\left[ \mathds{1}_{\{T_j \leq T_i - 2\}} -
     \frac{\mathds{1}_{\{T_j = T_i - 1\}}}
     {\exp\left(\sum_{k\colon T_k = T_i - 1} \theta^*_{si}\right)-1}
     \;\middle\vert\; \mathcal{F}_{T_j} \right] \\
	&= \mathbb{P}\left[T_i \geq T_j + 2 \;\middle\vert\; \mathcal{F}_{T_j}\right] -
	  \mathbb{E}\left[ \frac{\mathds{1}_{\{T_j = T_i - 1\}}}
		      {\exp\left(\sum_{k\colon T_k = T_i - 1} \theta^*_{ki}\right)-1}
		      \;\middle\vert\; \mathcal{F}_{T_j}\right]
\end{align*}
Considering the two terms above separately, we see that
\begin{align*}
  \mathbb{P}\left[T_i \geq T_j + 2 \;\middle\vert\; \mathcal{F}_{T_j}\right] = 
	\exp\left(-\sum_{k \colon T_k = T_j} \theta^*_{ki}\right)
\end{align*}
which follows from the fact that the probability that (active) $j$ failed to infect (susceptible) $i$ is equal
to the probability that all the nodes that were active at $T_j$ failed to infect $i$. For the second term, we have
\begin{align*}
	\mathbb{E}\left[ \frac{\mathds{1}_{\{T_j = T_i - 1\}}}
		      {\exp\left(\sum_{k\colon T_k = T_i - 1} \theta^*_{ki}\right)-1}
		      \;\middle\vert\; \mathcal{F}_{T_j}\right]
	& =\mathbb{E}\left[ \frac{\mathds{1}_{\{T_j = T_i - 1\}}}
		      {\exp\left(\sum_{k\colon T_k = T_j} \theta^*_{ki}\right)-1}
		      \;\middle\vert\; \mathcal{F}_{T_j}\right] \\
	& \stackrel{(\varsigma_1)}{=}\frac{1}{\exp\left(\sum_{k\colon T_k = T_j} \theta^*_{ki}\right)-1}
		    \mathbb{E}\left[ \mathds{1}_{\{T_j = T_i - 1\}}		      
		      \;\middle\vert\; \mathcal{F}_{T_j}\right] \\
	& \stackrel{(\varsigma_2)}{=}\frac{\left(1-\exp\left(-\sum_{k \colon T_k = T_j} \theta^*_{ki}\right)\right)
		\mathds{1}_{\{\exists k \in \V \mbox{ s.t. } T_k = T_j\}}}
	      {\exp\left(\sum_{k\colon T_k = T_j} \theta^*_{ki}\right)-1} \\
	& = \exp\left(-\sum_{k \colon T_k = T_j} \theta^*_{ki}\right) \mathds{1}_{\{\exists k \in \V \mbox{ s.t. } T_k = T_j\}}
\end{align*}
where $(\varsigma_1)$ follows from the fact that $\{k : T_k=T_j\}$ is measurable in $\mathcal{F}_{T_j}$ and
$(\varsigma_2)$ follows from the fact that $T_i=T_j+1$ if and only if at least one of the parents of $i$ were
active at $T_j$ and succeeded in infecting $i$. Combining the above two equations, we obtain
\begin{align}
\mathbb{E}\left[ \mathds{1}_{\{T_j \leq T_i - 2\}} -
     \frac{\mathds{1}_{\{T_j = T_i - 1\}}}
     {\exp\left(\sum_{k\colon T_k = T_i - 1} \theta^*_{si}\right)-1}
     \;\middle\vert\; \mathcal{F}_{T_j} \right]
=   \mathds{1}_{\{T_k \neq T_j \; \forall\; k \in \V\}} \mbox{ if } T_i > T_j \label{eqn:grad-conditional-2}
\end{align}
Combining \eqref{eqn:grad-conditional}, \eqref{eqn:grad-conditional-1} and \eqref{eqn:grad-conditional-2}
\begin{align}
  \grad_j L(\theta^*) = -\mathbb{P}\left[T_i > T_j; T_k \neq T_k \forall k \in \V\right]
\end{align}
\end{proof}

%%%%%%%%%%%%%%%%%%%%%%%

An easy corollary of Proposition \ref{prop:grad-characterization} is that if $j$ is a parent of
$i$, then the gradient with respect to $\theta_j$ is zero since the probability above
needs none of the parents of $i$ to be infected at the same time as $j$. On the other hand,
if $j$ is not a parent of $i$, the gradient is strictly negative since the probability on
the right hand side is strictly positive.
\begin{align}
  \grad_j L(\theta^*) &= 0 \mbox{ if } j \in \V \label{eqn:grad-neighbor-zero}\\
  \grad_j L(\theta^*) &< 0 \mbox{ if } j \notin \label{eqn:grad-nonneighb-nonzero}\V
\end{align}
We now state our concentration results.
For any $j$, let $\grad_j\widehat{L}(\theta)$ be the partial derivative of $\widehat{L}(\theta)$
with respect to $\theta_j$. For $j\in \V$, let 
\begin{align*}
  m_{1,j} \defas \left|\{u : t_j^u=t_i^u - 1 \;\&\;
  t_k^u \neq t_i^u - 1 \; \forall \; k \in \V\setminus j\}\right|
\end{align*}
be the number of cascades where $j$ is the sole infector of node $i$ and
\begin{align*}
  m_{2,j} \defas \left|\{u : t_j^u \leq t_i^u - 2 \}\right|
\end{align*}
be the number of cascades where $j$ is infected at least two time units before $i$.

%%%%%%%%%%%%%%%%%%%%%%%

\begin{lemma}\label{lemma:concentration-results}
  For $m > \frac{c}{p_{init}} \left ( \frac{1}{\alpha^7 \eta^2  p_{i,min}^2} \right )  d_i^2 \log \left(\frac{D_i}{\delta}\right)$, we have that
\begin{itemize}
  \item [(a)] $\left|\grad_j\widehat{L}(\theta^*)\right|<a$ for $j\in \V$ where
	  $a := \frac{\alpha^3 \eta p_{\textrm{init}}}{144 d}$
  \item [(b)] $\grad_j\widehat{L}(\theta^*)<-b$ for $j\notin \V$ where
	  $b := \frac{\alpha p_{\textrm{init}}}{16}$
  \item [(c)] $\xi_1 p_j^* < m_{1,j} < \overline{\xi}_1$ for $j \in \V$ where
	  $\xi_1 \defas \frac{c}{4}\log \frac{D}{\delta}$, $\overline{\xi}_1 \defas \frac{2c}{\alpha}\log\frac{D}{\delta}$ and $p_j^*\defas 1-\exp(-\theta_j^*)$
  \item [(d)] $\xi_2 < m_{2,j} < \overline{\xi}_2$ for $j \in \V$ where
	  $\xi_2 \defas \frac{c\alpha}{4}\log \frac{D}{\delta}$ and $\overline{\xi}_2 \defas \frac{2c}{\alpha}\log\frac{D}{\delta}$
\end{itemize}
with probability greater than $1-\delta$.
\end{lemma}
\begin{proof}
For simplicity of notation we denote the number of samples as $m = \frac{C\log\frac{D}{\delta}}{p_{\textrm{init}}}$
where $C = \frac{cd_i^2}{\alpha^7 \eta^2  p_{i,min}^2}$ and $D=D_i$.
We will first prove (c).
First, we note the following bounds for independent Bernoulli random variables $X_l$ where $\mu$ is the mean of the sum of
$X_l$.
\begin{align}
  \mathbb{P}\left[\sum_l X_l < (1-\kappa)\mu\right] < \left(\frac{\exp(-\kappa)}{(1-\kappa)^{(1-\kappa)}}\right)^{\mu} \label{eqn:chernoff-lower}\\
  \mathbb{P}\left[\sum_l X_l > (1+\kappa)\mu\right] < \left(\frac{e^{\frac{\kappa}{1+\kappa}}}{1+\kappa}\right)^{(1+\kappa)\mu} \label{eqn:chernoff-upper}
\end{align}
So as to be able to use the above inequalities, we first establish bounds on the expected value of $m_{1,j}$.
\begin{align*}
  \mathbb{E}\left[m_{1,j}\right] &\geq m p_{\textrm{init}} (1-p_{\textrm{init}})^d p_j^*
      \geq 2\xi_1 p_j^*
\end{align*}
where the bound uses the probability that $j$ is infected at time $0$ and neither $i$ nor
any of its other neighbors are infected at time $0$ and $j$ infects $i$ at time $1$.
Similarly, we have
\begin{align*}
  \mathbb{E}\left[m_{1,j}\right] &\leq m \mathbb{P}\left[T_j < \infty\right] 
      \leq \frac{\overline{\xi}_1}{2}
\end{align*}
where we use Lemma \ref{lem:corrdecay-probbound}.
Now applying \eqref{eqn:chernoff-lower} to $m_{1,j}$ we obtain
\begin{align*}
  \mathbb{P}\left[m_{1,j} < (1-\frac{1}{2}) 2\xi_1 p_j^*\right] <
	\left(\frac{\exp\left(-\frac{1}{2}\right)}
	      {\left(\frac{1}{2}\right)^{\frac{1}{2}}}\right)^{2\xi_1 p_j^*} 
    < \frac{\delta}{8D}
\end{align*}
Similarly applying \eqref{eqn:chernoff-upper} to $m_{1,j}$ gives us
\begin{align*}
  \mathbb{P}\left[m_{1,j} > (1+1) \frac{\overline{\xi}_1}{2}\right] <
	\left(\frac{\sqrt{e}}{2}\right)^{\overline{\xi}_1} 
    < \frac{\delta}{8D}
\end{align*}
This proves (c). The proof of (d) is similar.

We will now prove (a). Fix any $j \in \V$. Let $\U_j = \{u\in \U:T_j^u < \infty\}$.
Since $\mathbb{E}\left[\left|\U_j\right|\right] \geq p_{\textrm{init}} m = C \log \frac{D}{\delta}$, using
\eqref{eqn:chernoff-lower}, we obtain
\begin{align}\label{eqn:Uj-lowerbound}
  \mathbb{P}\left[|\U_j| < \frac{C\log \frac{D}{\delta}}{2}\right] < \frac{\delta}{16D}
\end{align}
Similarly since $\mathbb{E}\left[\left|\U_j\right|\right] \leq \frac{p_{\textrm{init}}}{\alpha} m = \frac{C}{\alpha} \log \frac{D}{\delta}$,
using \eqref{eqn:chernoff-upper}, we obtain
\begin{align}\label{eqn:Uj-upperbound}
  \mathbb{P}\left[|\U_j| > \frac{2C\log \frac{D}{\delta}}{\alpha}\right] < \frac{\delta}{16D}
\end{align}
Define the random variable
\begin{align*}
  Z_j = -\mathrm{1}_{\{T_j \leq T_i - 2\}} + \frac{\mathrm{1}_{\{T_j=T_i-1\}}}{\exp\left(\sum_{k:T_k=T_i-1}\theta_k^*\right)-1}
\end{align*}
Note that we have the following absolute bound on $Z_j$
\begin{align}\label{eqn:Zj-absbound}
  |Z_j| < 1+\frac{1}{\exp(\theta^*_{j})-1} = \frac{1}{p^*_{j}}
\end{align}
where $p^*_j=1-\exp\left(-\theta_j\right)$ and also
\begin{align*}
\grad_j \widehat{L}(\theta^*) &= \frac{1}{m} \sum_{u\in \U} Z_j^u
	= \frac{1}{m} \sum_{u \in \U_j} Z_j^u
\end{align*}
where $Z_j^u$ is the realization of $Z_j$ on infection $u$.
% Note that since $\theta^*$ is the optimum of $L(\theta)$,
% we have $\grad_j L(\theta^*) = \mathbb{E}\left[Z_j\right]=0$.
\begin{align*}
  \mathbb{P}\left[\left|\grad_j\widehat{L}(\theta^*)\right|\geq a\right]
  = \mathbb{P}\left[\frac{1}{m}\left|\sum_{u\in \U} Z_j^u\right|\geq a\right]
  = \mathbb{P}\left[\left|\sum_{u\in \U} Z_j^u\right|\geq ma\right] \nonumber
\end{align*}
At this point we could apply Azuma-Hoeffding inequality to bound the above probability. However, the scaling factor in the
exponent will be $ma^2$ which gives us an extra $p_{\textrm{init}}$. To avoid this, we bound the above quantity as follows:
\begin{align}
  &\mathbb{P}\left[\left|\sum_{u\in \U} Z_j^u\right|\geq ma\right] \nonumber \\
%   &\leq\sum_{s=0}^{m} \mathbb{P}\left[\left|\U_j\right|=s;\left|\sum_{u\in \U} Z_j^u \right|\geq ma\right] \nonumber\\
  &\leq \mathbb{P}\left[|\U_j| > \frac{2C\log \frac{D}{\delta}}{\alpha} \mbox{ or } |\U_j| < \frac{C\log \frac{D}{\delta}}{2}\right]
% +\mathbb{P}\left[U_j > \frac{2c\log \frac{D}{\delta}}{\alpha}\right] \nonumber\\
	+ \sum_{s=\frac{C\log \frac{D}{\delta}}{2}}^{\frac{2C\log \frac{D}{\delta}}{\alpha}}
	    \mathbb{P}\left[\left|\U_j\right|=s;\left|\sum_{u\in \U} Z_j^u \right|\geq ma\right] \nonumber\\
%   &\leq\frac{\delta}{8D} \nonumber \\
%   &\;\;+\sum_{s=\frac{C\log \frac{D}{\delta}}{2}}^{\frac{2C\log \frac{D}{\delta}}{\alpha}}
%     \sum_{U_j : |U_j|=s}\mathbb{P}\left[\U_j=U_j;\left|\sum_{u\in U_j} Z_j^u \right|\geq ma \right] \nonumber\\
  &\stackrel{(\varsigma_1)}{\leq}\frac{\delta}{8D}+
  \sum_{s=\frac{C\log \frac{D}{\delta}}{2}}^{\frac{2C\log \frac{D}{\delta}}{\alpha}}
    \sum_{U_j : |U_j|=s}\mathbb{P}\left[\U_j=U_j\right]\mathbb{P}\left[\left|\sum_{u\in U_j} Z_j^u \right|\geq ma
      \; \middle \vert \; \U_j = U_j\right] \label{eqn:grad-neighb-concentration1}
\end{align}
where $U_j$ varies over all the subsets of $\U$ and $(\varsigma_1)$ follows from \eqref{eqn:Uj-lowerbound}
and \eqref{eqn:Uj-upperbound}.
Focusing on the last term, we first note that $Z_j^u$ are still independent random variables for $u\in U_j$.
Since $\mathbb{E}\left[Z_j\right]=0$ from \eqref{eqn:grad-neighbor-zero}, we can apply Azuma-Hoeffding inequality and using \eqref{eqn:Zj-absbound} we obtain
\begin{align}
  \mathbb{P}\left[\left|\sum_{u\in U_j} Z_j^u \right|\geq ma \; \middle \vert \; \U_j = U_j, |U_j|=s\right]
  \leq 2\exp\left(\frac{-(ma)^2}{2s\left(\frac{1}{p^*_j}\right)^2}\right)
%   &\stackrel{(\varsigma_1)}{\leq} 2\exp\left(\frac{-(\alpha map^*_j)^2}{4C\log\frac{D}{\delta}}\right)
  < \frac{\delta}{16D} \label{eqn:grad-neighb-concentration2}
\end{align}
where $(\varsigma_1)$ follows from the fact that $s \leq \frac{2C\log \frac{D}{\delta}}{\alpha}$.
The proof of (b) is on the same lines after noting that for any $j\notin \V$,
\begin{align}
  \mathbb{E}\left[Z_j\right]=\grad_j L(\theta^*) &\stackrel{(\varsigma_1)}{=} -\mathbb{P}\left[T_i > T_j ; T_j \neq T_k \;\forall \; k \in \V\right] \nonumber\\
	&\stackrel{(\varsigma_2)}{<} -p_{\textrm{init}}\left(1-p_{\textrm{init}}\right)^{d+1}
	\stackrel{(\varsigma_3)}{<} -\frac{p_{\textrm{init}}}{2} \label{eqn:Zj-nonneighb-upperbound}
\end{align}
where $(\varsigma_1)$ follows from Proposition \ref{prop:grad-characterization},
$(\varsigma_2)$ follows from the fact that the probability when $j$ is infected before $i$ and none of the
parents of $i$ are infected at the same time can be lower bounded by the case where $j$ is infected at time $0$
and neither $i$ nor any of its parents are infected at time $0$. $(\varsigma_3)$ follows from the assumption
that $p_{\textrm{init}} < \frac{1}{2d}$ and hence $(1-p_{\textrm{init}})^{d+1} > \frac{1}{2}$. Using
\eqref{eqn:Zj-nonneighb-upperbound} and Lemma \ref{lem:corrdecay-probbound}, we obtain
\begin{align}
  \mathbb{E}\left[Z_j \; \middle \vert \; T_j < \infty\right] &=
      \frac{\mathbb{E}\left[Z_j\right] - \mathbb{E}\left[Z_j\mathrm{1}_{\{T_j = \infty\}}\right]}{\mathbb{P}\left[T_j < \infty\right]} %\nonumber\\
%       &< \frac{-\frac{p_{\textrm{init}}}{2}}{\frac{p_{\textrm{init}}}{\alpha}}
      < = \frac{-\alpha}{2} \label{eqn:Zj-Tjfinite-nonneighb-upperbound}
\end{align}
Using \eqref{eqn:grad-neighb-concentration1} it suffices to show that
\begin{align*}
  \mathbb{P}\left[\sum_{u\in U_j} Z_j^u \geq -mb \; \middle \vert \; \U_j = U_j, |U_j|=s\right]
  < \frac{\delta}{16D} \label{eqn:grad-nonneighb-concentration2}
\end{align*}
for $\frac{C\log\frac{D}{\delta}}{2}\leq s \leq \frac{2C\log \frac{D}{\delta}}{\alpha}$. An application of Azuma-Hoeffding inequality gives us the required bound as follows.
\begin{align*}
  &\mathbb{P}\left[\sum_{u\in U_j} Z_j^u \geq -mb \; \middle \vert \; \U_j = U_j, |U_j|=s\right] \\
%   &=\mathbb{P}\left[\sum_{u\in U_j} Z_j^u \geq \frac{-C\alpha\log\frac{D}{\delta}}{16} \; \middle \vert \; \U_j = U_j, |U_j|=s\right] \\
  &\stackrel{(\varsigma_1)}{=}\mathbb{P}\left[\sum_{u\in U_j} Z_j^u - s \mathbb{E}\left[Z_j\right]\geq
    \frac{-C\alpha\log\frac{D}{\delta}}{16} - s \mathbb{E}\left[Z_j\right] \; \middle \vert \; \U_j = U_j, |U_j|=s\right] \\
%   &\stackrel{(\varsigma_2)}{\leq} \mathbb{P}\left[\sum_{u\in U_j} Z_j^u - s \mathbb{E}\left[Z_j\right]\geq \right. \\
%     &\;\;\;\;\;\;\;\;\;\;\;\left.\frac{-C\alpha\log\frac{D}{\delta}}{16} - \frac{C\log\frac{D}{\delta}}{2} \frac{-\alpha}{2} \; \middle \vert \; \U_j = U_j, |U_j|=s\right] \\
  &\stackrel{(\varsigma_2)}{\leq} \mathbb{P}\left[\sum_{u\in U_j} Z_j^u - s \mathbb{E}\left[Z_j\right]\geq %\right.
%     &\;\;\;\;\;\;\;\;\;\;\; \left.
\frac{C\alpha\log\frac{D}{\delta}}{8} \; \middle \vert \; \U_j = U_j, |U_j|=s\right] \\
  &\stackrel{(\varsigma_3)}{\leq} \exp\left(\frac{\left(\frac{C\alpha\log\frac{D}{\delta}}{8}\right)^2}{2\left(\frac{2C\log\frac{D}{\delta}}{\alpha}\right)\left(\frac{1}{p^*_j}\right)^2}\right)
  \leq \frac{\delta}{16D}
\end{align*}
where $(\varsigma_1)$ follows by subtracting $s\mathbb{E}\left[Z_j\right]$ from both sides of the inequality for which we are bounding the probability,
$(\varsigma_2)$ follows from the fact that $s \geq \frac{C\log\frac{D}{\delta}}{2}$ and \eqref{eqn:Zj-Tjfinite-nonneighb-upperbound} and
$(\varsigma_3)$ is an application of the Azuma-Hoeffding inequality using \eqref{eqn:Zj-absbound} and the fact that
$s \leq \frac{2C\log \frac{D}{\delta}}{\alpha}$.
\end{proof} 

\vspace{0.1in}

%%%%%%%%%%%%%%%%%%%%%%%

\begin{lemma}\label{lemma:thetaneighb-upperbound}
When (a)-(d) in Lemma \ref{lemma:concentration-results} hold, \,\,  $\max_{j\in \V} \widehat{\theta}_j < \frac{\overline{\xi}_1}{\xi_2}$
\end{lemma}
\begin{proof}
  Let $k=\argmax_{j\in\V} \widehat{\theta}_j$. If $\widehat{\theta}_k=0$, we are done.
So assume $\widehat{\theta}_k>0$. By the optimality of $\widehat{\theta}$, we see that
\begin{equation}\label{eqn:thetaneighb-ub1}
  \grad_k\widehat{L}(\widehat{\theta})  = 0
\end{equation}
On the other hand, we have
\begin{align}
  \grad_k\widehat{L}(\widehat{\theta})
  &= \frac{1}{m} \left(-m_{2,k} +
	\sum_u \mathrm{1}_{\{t_i^u < \infty\}}
	  \left(\exp\left(\displaystyle\sum_{j:t_j^u =t_i^u -1}\widehat{\theta}_j\right)-1\right)^{-1}\right) \nonumber \\
  &\stackrel{(\varsigma_1)}{\leq} \frac{1}{m} \left(-m_{2,k} + \frac{1}{\exp(\widehat{\theta}_k)-1}m_{1,k}\right) \nonumber \\
  &\stackrel{(\varsigma_2)}{\leq} \frac{1}{m} \left(-\xi_2 + \frac{1}{\exp(\widehat{\theta}_k)-1}\overline{\xi}_1\right)
  \leq \frac{1}{m} \left(-\xi_2 + \frac{1}{\widehat{\theta}_k}\overline{\xi}_1\right) \label{eqn:thetaneighb-ub2}
\end{align}
where $(\varsigma_1)$ follows from the definition of $m_{1,k}$ and the fact that
on the infections corresponding to $m_{1,k}$, we have
\begin{align*}
\displaystyle\sum_{j:t_j^u =t_i^u -1}\widehat{\theta}_j \geq \widehat{\theta}_k
\end{align*}
and $(\varsigma_2)$ follows from Lemma \ref{lemma:concentration-results}.
Putting \eqref{eqn:thetaneighb-ub1} and \eqref{eqn:thetaneighb-ub2} together, we obtain the result.
\end{proof}

\vspace{0.1in}

%%%%%%%%%%%%%%%%%%%%%%%

\begin{lemma}\label{lemma:thetanonneighb-upperbound}
When (a)-(d) in Lemma \ref{lemma:concentration-results} hold, \,\, $\sum_{j\notin \V} \widehat{\theta}_j ~ \leq ~ \frac{ad}{b}\left(\frac{\overline{\xi}_1}{\xi_2}+\log \frac{1}{\alpha}\right) ~ < ~ \eta$
\end{lemma}
\begin{proof}
Since $\widehat{L}(\theta)$ is concave, the subgradient condition at $\theta^*$ gives us the following
\begin{align}
  \widehat{L}(\widehat{\theta}) - \widehat{L}(\theta^*) \nonumber
  &\leq \left\langle \grad\widehat{L}(\theta^*), \widehat{\theta} - \theta^* \right\rangle \nonumber\\
  &\stackrel{(\varsigma_1)}{=} \left\langle \grad_{{\V}^c}\widehat{L}(\theta^*), \widehat{\theta}_{{\V}^c}\right\rangle
    + \left\langle \grad_{\V}\widehat{L}(\theta^*), \widehat{\theta}_{\V} -\theta_{\V}^*\right\rangle \nonumber\\
  &\stackrel{(\varsigma_2)}{\leq} -b ||\widehat{\theta}_{{\V}^c}||_1 + a ||\widehat{\theta}_{\V}-\theta_{\V}^*||_1 \nonumber\\
%   &\leq -b ||\widehat{\theta}_{{\V}^c}||_1 + a \left(||\widehat{\theta}_{\V}||_1+||\theta_{\V}^*||_1\right) \nonumber\\
  &\leq -b ||\widehat{\theta}_{{\V}^c}||_1 + ad \left(||\widehat{\theta}_{\V}||_{\infty}+||\theta_{\V}^*||_{\infty}\right) \label{eqn:thetanonneighb-ub1}
\end{align}
where $(\varsigma_1)$ follows from the fact that $\theta_{\V^c}^*=0$ and $(\varsigma_2)$ follows from the fact that
$\widehat{\theta}>0$ and Lemma \ref{lemma:concentration-results}. The optimality of $\widehat{\theta}$ gives us
\begin{align}\label{eqn:thetanonneighb-ub2}
  \widehat{L}(\widehat{\theta}) - \widehat{L}(\theta^*) \geq 0
\end{align}
Finally we have the following bound on $||\theta_{\V}^*||_{\infty}$:
\begin{align}\label{eqn:thetamaxbound}
  \theta^*_j = -\log\left(1-p^*_j\right) \leq \log\frac{1}{\alpha}
\end{align}
Using \eqref{eqn:thetanonneighb-ub1}, \eqref{eqn:thetanonneighb-ub2}, \eqref{eqn:thetamaxbound} and Lemma \ref{lemma:thetaneighb-upperbound} proves the first inequality, that $\sum_{j\notin \V} \widehat{\theta}_j ~ \leq ~ \frac{ad}{b}\left(\frac{\overline{\xi}_1}{\xi_2}+\log \frac{1}{\alpha}\right)$. The second inequality, that $\frac{ad}{b}\left(\frac{\overline{\xi}_1}{\xi_2}+\log \frac{1}{\alpha}\right) < \eta$, is easy to see. 
\end{proof}

\vspace{0.1in}

%%%%%%%%%%%%%%%%%%%%%%%

\begin{lemma}\label{lemma:thetaneighb-lowerbound}
When (a)-(d) in Lemma \ref{lemma:concentration-results} hold, for every $j \in \V$ we have that $\widehat{\theta}_j > \log \left(1+\frac{p_j^*\xi_1}{\overline{\xi}_2}\right)-\eta$ where
$p_j^*=1-\exp(\theta_j^*)$.
\end{lemma}
\begin{proof}
Since $\widehat{\theta}_j\geq 0$, by the optimality of $\widehat{\theta}$ we have
\begin{align}\label{eqn:thetaneighb-lb1}
  \grad_j \widehat{L}(\widehat{\theta}) \leq 0
\end{align}
On the other hand, we have the following bound on the gradient
\begin{align}
  \grad_j\widehat{L}(\widehat{\theta})
  &= \frac{1}{m} \left(-m_{2,j} -
	\sum_u \mathrm{1}_{\{t_i^u < \infty\}}
	  \left(\exp\left(\displaystyle\sum_{k:t_k^u =t_i^u -1}\widehat{\theta}_k\right)-1\right)^{-1}\right) \nonumber \\
  &\stackrel{(\varsigma_1)}{\geq} \frac{1}{m} \left(-m_{2,j} + \frac{1}{\exp\left(\widehat{\theta}_j+||\widehat{\theta}_{{\V}^c}||_1\right)-1}m_{1,k}\right) \nonumber \\
  &\stackrel{(\varsigma_2)}{\geq} \frac{1}{m} \left(-\overline{\xi}_2 + \frac{1}{\exp\left(\widehat{\theta}_j+||\widehat{\theta}_{{\V}^c}||_1\right)-1}p_j^*\xi_1\right) \label{eqn:thetaneighb-lb2}  
\end{align}
where $(\varsigma_1)$ follows from the fact that on the infections corresponding to
$m_{1,k}$, we have 
\begin{align*}
  \displaystyle\sum_{k:t_k^u =t_i^u -1}\widehat{\theta}_k \leq \widehat{\theta}_j+||\widehat{\theta}_{{\V}^c}||_1
\end{align*}
and $(\varsigma_2)$ follows from Lemma \ref{lemma:concentration-results}.
Combining \eqref{eqn:thetaneighb-lb1}, \eqref{eqn:thetaneighb-lb2} and Lemma \ref{lemma:thetanonneighb-upperbound}
gives us the result.
\end{proof}
Thus we see that if the true parameter $p^*_{ji} > \frac{8}{\alpha} (e^{2\eta} - 1)$, then $\widehat{\theta}_j > \eta$ and thus will be in the estimated
neighborhood $\widehat{\mathcal{N}}_i$. This completes the proof of Theorem \ref{thm:samplecomplexity_maxLL}.

%%%%%%%%%%%%%%%%%%%%%%%%%%%%%%%%%%%%%%%%%%%%%%%%%%%%%%%%%%%%%%%%%%%%%%%%%%%%%%%%%%%%%%%%%%

% \newpage
\section{Greedy algorithm}

\subsection{Proof of Theorem \ref{thm:greedy-guarantees}}

% We also assume for simplicity that all the edge probabilities $p_{ki}$ are the same and are equal to $p$.
To simplify notation, we again denote $\V_i$ by $\V$, $\mS_i$ by $\mS$ and so on.
From Lemma \ref{lem:corrdecay-probbound}, we have that for every node $j$,
\begin{align*}
  \mathbb{P}\left[T_j < \infty\right] < \frac{p_{\textrm{init}}}{\alpha}
\end{align*}
Since the graph is a tree, for every node $j$ there exists a unique (undirected) path between $i$ and $j$.
All the nodes on this path are said to be ancestors of $j$.
Consider a node $j \in \mS \setminus \V$. Let $k \in \V$ be the ancestor of $j$ on this path. Then we have that
\begin{align*}
  \mathbb{P}\left[T_j \neq T_k; T_k = T_i - 1\right] \geq p_{\textrm{init}} \left(1-p_{\textrm{init}}\right)^2 p_{\textrm{min}}
\end{align*}
If $l\in \V$ but is not an ancestor of $j$ then
\begin{align*}
  \mathbb{P}\left[T_j = T_l = T_i - 1\right] 
%       &< \mathbb{P}_{G\setminus i}\left[T_j<\infty]\mathbb{P}_{G\setminus i}[T_l<\infty\right] \\
      &<\mathbb{P}\left[T_j<\infty]\mathbb{P}[T_l<\infty\right]
      < \left(\frac{p_{\textrm{init}}}{\alpha}\right)^2
\end{align*}
since $T_j$ and $T_l$ are independent conditioned on $T_j, T_l < T_i$.
For any event $A$ that depends on the infection times, let $N(A)$ denote the number of cascades in $\U$ in which event $A$ has occurred.
% Note that the number of samples is $m = \frac{c \log D}{p_{\textrm{init}}}$.
Using \eqref{eqn:chernoff-lower} and \eqref{eqn:chernoff-upper},
we have the following bounds on probabilities of error events:
% $1 - \left(\frac{2}{e}\right)^{\frac{cp_{\textrm{min}}\left(1-p_{\textrm{init}}\right)\log D}{2}}
% -\left(\frac{8edp_{\textrm{init}}}{\alpha^2 p_{\textrm{min}}}\right)^{\frac{cp_{\textrm{min}}\log D}{8d}}
% -\left(\frac{8edp_{\textrm{init}}}{\alpha^2 p_{\textrm{min}}}\right)^{\frac{cp_{\textrm{min}}\log D}{8d}}
% -\left(\frac{2}{e}\right)^{\frac{cp_{\textrm{min}}\left(1-p_{\textrm{init}}\right)^2\log D}{2}}$,
\begin{eqnarray}
  \mathbb{P}\left[N\left(T_k = T_i - 1\right) \leq \left(1-\frac{1}{2}\right)mp_{\textrm{min}}p_{\textrm{init}}\left(1-p_{\textrm{init}}\right)\right] 
	&< \left(\frac{2}{e}\right)^{\frac{mp_{\textrm{min}}p_{\textrm{init}}\left(1-p_{\textrm{init}}\right)}{2}} \label{eq:neighb_lowbound_prob} \nonumber \\
  \mathbb{P}\left[N\left(T_k = T_l = T_i - 1\right) \geq \frac{m p_{\textrm{init}}p_{\textrm{min}}}{8d}\right]
	&< \left(\frac{em\left(\frac{p_{\textrm{init}}}{\alpha}\right)^2}{\left(\frac{mp_{\textrm{init}}p_{\textrm{min}}}{8d}\right)}\right)^
	    {\frac{mp_{\textrm{init}}p_{\textrm{min}}}{8d}} \label{eq:neighb_upbound_prob} \nonumber \\
  \mathbb{P}\left[N\left(T_j = T_l = T_i - 1\right) \geq \frac{m p_{\textrm{init}}p_{\textrm{min}}}{8d}\right]
	&< \left(\frac{em\left(\frac{p_{\textrm{init}}}{\alpha}\right)^2}{\left(\frac{mp_{\textrm{init}}p_{\textrm{min}}}{8d}\right)}\right)^
	    {\frac{mp_{\textrm{init}}p_{\textrm{min}}}{8d}} \label{eq:nonneighb_upbound_prob} \nonumber \\
  \mathbb{P}\left[N\left(T_j \neq T_k; T_k = T_i - 1\right) \leq \left(1-\frac{1}{2}\right)mp_{\textrm{init}}p_{\textrm{min}}\left(1-p_{\textrm{init}}\right)^2\right]
	&< \left(\frac{2}{e}\right)^{mp_{\textrm{init}}p_{\textrm{min}}\left(1-p_{\textrm{init}}\right)^2}\label{eq:diff_lowbound_prob}
\end{eqnarray}
where $k,l \in \V$ and $j\notin \V$ such that $k$ is an ancestor of $j$.
Substituting the value of $m$ from the statement of Theorem \ref{thm:greedy-guarantees} and recalling the assumption on $p_{\textrm{init}}$,
we see that with probability greater than $1-\delta$, we have
\begin{align}
  N\left(T_k = T_i - 1\right) &> \frac{cd\left(1-p_{\textrm{init}}\right)\log\frac{D}{\delta}}{2} \label{eq:neighb_lowbound}\\
  N\left(T_k = T_l = T_i - 1\right) &< \frac{c \log \frac{D}{\delta}}{8} \label{eq:neighb_upbound}\\
  N\left(T_j = T_l = T_i - 1\right) &< \frac{c \log \frac{D}{\delta}}{8} \label{eq:nonneighb_upbound}\\
  N\left(T_j \neq T_k; T_k = T_i - 1\right) &> \frac{cd\left(1-p_{\textrm{init}}\right)^2 \log \frac{D}{\delta}}{2} \label{eq:diff_lowbound}
\end{align}

Note that the assumption on $p_{\textrm{init}}$ also yields an upper bound of $\frac{1}{16}$ on $p_{\textrm{init}}$.
Now we will show that under the above conditions, Algorithm \ref{algo:greedy} recovers the original graph exactly. Suppose in iteration
$s$, the neighborhood is $s-1$ of the correct parents and there is atleast one $k \in \V$, not in the current neighborhood. Let the
current set of infections be $U$. Then from \eqref{eq:neighb_lowbound} and \eqref{eq:neighb_upbound}, we see that
\begin{align*}
  N_U\left(T_k = T_i - 1\right) &> \frac{cd\left(1-p_{\textrm{init}}\right)\log \frac{D}{\delta}}{2} - d \;\frac{c \log \frac{D}{\delta}}{8} \\
	&= \frac{cd\log \frac{D}{\delta}}{8}\left(4\left(1-p_{\textrm{init}}\right)-1\right) 
	> 0
\end{align*}
So there is atleast one node that will be added to the neighborhood. Now consider any $j \notin \V$. If the ancestor of $j$ that
is a parent of $i$ has already been added to the neighborhood list, then from \eqref{eq:nonneighb_upbound}
\begin{align*}
  N_U\left(T_j = T_i - 1\right) &< d \; \frac{c \log \frac{D}{\delta}}{8} \\
% 		  &< \frac{cp_{\textrm{min}}\log D}{8} \\
		  &< \left(4\left(1-p_{\textrm{init}}\right)-1\right)\frac{cd\log \frac{D}{\delta}}{8} \\
		  &< N_U\left(T_k = T_i - 1\right)
\end{align*}
Suppose the ancestor of $j$ that is a parent of $i$ has not yet been added to the neighborhood of $i$. Without loss of generality, let $k$
be the ancestor of $j$. Then,
\begin{align*}
  &N_U\left(T_k = T_i - 1\right) - N_U\left(T_j = T_i - 1\right) \\
% &= N_U\left(T_j \neq T_k; T_k = T_i - 1\right) - N_U\left(T_j \neq T_k; T_j = T_i - 1\right) \\
&= N_U\left(T_j \neq T_k; T_k = T_i - 1\right) \\
&\;\;\; -N_U\left(T_j = T_l = T_i - 1 \colon l \neq k, l \in S\right) \\
% &> N\left(T_j \neq T_k; T_k = T_i - 1\right) \\
% &\;\;\;- N\left(T_k = T_l = T_i - 1 \colon l \neq k, l \in S\right)\\
% &\;\;\;- N\left(T_j = T_l = T_i - 1 \colon l \neq k, l \in S\right) \\
&> \frac{cd\left(1-p_{\textrm{init}}\right)^2 \log \frac{D}{\delta}}{2} - 2d\; \frac{c \log \frac{D}{\delta}}{8} \\
&= \frac{c d \log \frac{D}{\delta}}{4}\left(2\left(1-p_{\textrm{init}}\right)^2 - 1\right)
> 0
\end{align*}

Applying union bound over all nodes in the superneighborhood, we can conclude that all nodes in the superneighborhood satisfy
\eqref{eq:neighb_lowbound}, \eqref{eq:neighb_upbound}, \eqref{eq:nonneighb_upbound} and \eqref{eq:diff_lowbound} with probability
greater than $1-\delta$. This proves Theorem \ref{thm:greedy-guarantees}.

%%%%%%%%%%%%%%%%%%%%%%%%%%%%%%%%%%%%%%%%%%%%%%%%%%%%%%%%%%%%%%%%%%%%%%%%%%%%%%%%%%%%%%%%%%

\section{Lower Bounds}

\vspace{0.05in}
\subsection{Proof of Lemma \ref{lem:entropy_upperbound}}
Recall from Lemma \ref{lem:corrdecay-probbound} that $\mathbb{P}\left[T_i = t\right] \leq \left(1-\alpha\right)^{t-1} p_{\textrm{init}}$. The proof just involves using this to bound $H(T_i)$.
Since $p_{\textrm{init}} < \frac{1}{e}$, we have the following
\begin{eqnarray*}
 H(T_i) 
	  &= & -\sum_{t=1}^n \mathbb{P}\left[T_i = t\right] \log \mathbb{P}\left[T_i = t\right] \\
		& &\;\;\;- \mathbb{P}\left[T_i = \infty\right] \log \mathbb{P}\left[T_i = \infty\right] \\
	 &\leq & -\sum_{t=1}^n \left(1-\alpha\right)^{t-1}p_{\textrm{init}} \log \left(1-\alpha\right)^{t-1}p_{\textrm{init}} \\
	      & & \;\;\;- \left(1-\frac{p_\textrm{init}}{\alpha}\right)\log\left(1-\frac{p_\textrm{init}}{\alpha}\right)\\
% 	 &\leq & -\sum_{t=1}^{\infty} \left(1-\alpha\right)^{t-1}p_{\textrm{init}} \log p_{\textrm{init}} \\
% 		 & & \;\;\; -\sum_{t=1}^{\infty} (t-1)\left(1-\alpha\right)^{t-1}p_{\textrm{init}} \log \left(1-\alpha\right) \\
% 	 & &\;\;\;\;     - \left(1-\frac{p_\textrm{init}}{\alpha}\right)\log\left(1-\frac{p_\textrm{init}}{\alpha}\right)\\
% 	 &= & -p_{\textrm{init}} \log p_{\textrm{init}} \frac{1}{1-\alpha}
% 		  -p_{\textrm{init}} \log \left(1-\alpha\right)\frac{1-\alpha}{\alpha^2} \\
% 	 & & \;\;\;\;     - \left(1-\frac{p_\textrm{init}}{\alpha}\right)\log\left(1-\frac{p_\textrm{init}}{\alpha}\right)\\
	 &\stackrel{(\varsigma_1)}{\leq} & \frac{p_{\textrm{init}}}{1-\alpha} \left(\log \frac{1}{p_{\textrm{init}}} +
		  \left(\frac{1-\alpha}{\alpha}\right)^2\log \frac{1}{1-\alpha}\right) \\
		& &\;\;\;  - \left(1-\frac{p_\textrm{init}}{\alpha}\right)\log\left(1-\frac{p_\textrm{init}}{\alpha}\right)
\end{eqnarray*}
where $(\varsigma_1)$ follows from some algebraic manipulations.

%%%%%%%%%%%%%%%%%%%%%%%%%%%%%%%%%%%%%%%%%%%%%%%%%%%%%%%%%%%%%%%%%%%%%%%%%%%%%%%%%%%%%%%%%%

\section{Generalized Independent Cascade Model}

\vspace{0.05in}
\subsection{Proof of Prop. \ref{prop:likelihood-genIC}}
Defining
\begin{align*}
x_i^0(\tau) = 
  \begin{cases}
    0 & \mbox{ if } \tau < t_i \\
    1 & \mbox{ if } \tau \geq t_i
  \end{cases}
\end{align*}
and proceeding as in the proof of Proposition \ref{prop:likelihood}, we obtain
\begin{align*}
 \mathbb{P}_{\theta}\left[T=t\right] ~
	= ~ \mathbb{P}_{\theta}\left[X(0) = x^0(0)\right] ~ \times ~
	\prod_{\tau = 1}^n \mathbb{P}_{\theta}\left[X(\tau) = x^0(\tau) \middle \vert X(0:\tau-1) = x^0(0:\tau-1)\right]
\end{align*}
where $X(0:\tau)$ denotes the (joint) values of the vectors $X(0),\cdots,X(\tau)$.
Now, $\mathbb{P}_{\theta}\left[X(0) = x^0(0)\right] = p_{\textrm{init}}^s\left(1-p_{\textrm{init}}\right)^{n-s}$.
Also,
\begin{align*}
  \mathbb{P}_{\theta}\left[X(\tau) = x^0(\tau) \middle \vert X(0:\tau-1) = x^0(0:\tau-1)\right]
  ~ = ~ \prod_{i\in V} \mathbb{P}_{\theta}\left[X_i(\tau) = x_i^0(\tau) \middle \vert X(0:\tau-1) = x^0(0:\tau-1)\right]
\end{align*}
because each node gets infected independently from each of its currently active neighbors. Thus we have that
\begin{align}
  \mathbb{P}\left[T=t\right]
      = p_{\textrm{init}}^s\left(1-p_{\textrm{init}}\right)^{n-s} 
	    \prod_{i\in V} \left(\prod_{\tau=1}^n b_i(\tau)\right)\label{eqn:maxLL-probproduct-genIC}
\end{align}
where $b_i(\tau) = \mathbb{P}_{\theta}\left[X_i(\tau) = x_i^0(\tau) \middle \vert X(0:\tau-1) = x^0(0:\tau-1)\right]$.
It is clear that for $\tau > t_i$, $b_i(\tau)=1$. For $\tau=t_i$, $b_i(\tau)$ is the probability that at least one
of the parents $j$ of $i$ infected before $t_i$ infected node $i$ at time $t_i$ given that $j$ did not infect $i$
before $t_i$. Thus,
\begin{align}
  b_i(t_i) &= 1 - \prod_{j:t_j<t_i} \frac{1-\sum_{r\in[t_i]} p_{ji}^r}{1-\sum_{r\in[t_i-1]} p_{ji}^r} \nonumber \\
	  &= 1 - \prod_{j:t_j<t_i} \exp\left(-\theta_{ji}^{t_i-t_j}\right)\label{eqn:maxLL-infectionpart-genIC}
\end{align}
Finally, for each $\tau < t_i$, $b_i(\tau)$ is the probability that active nodes at time $\tau-1$ failed to infect node $i$.
The set of all nodes that were active but failed to infect susceptible node $i$ is $\{j:t_j \leq t_i-2\}$. Each such node
$j$ failed to infect $i$ for $t_i-t_j-1$ time slots. So we have
\begin{align}
  \prod_{\tau<t_i} b_i(\tau) &= \prod_{j:t_j \leq t_i-2} \left(1-\sum_{r\in[t_i-t_j-1]}p_{ji}^r\right) \nonumber\\
	&= \prod_{j:t_j \leq t_i-2} \prod_{r\in [t_i-t_j-1]} \exp\left(-\theta_{ji}^r\right) \label{eqn:maxLL-failedinfectionpart-genIC}
\end{align}
Putting \eqref{eqn:maxLL-probproduct-genIC}, \eqref{eqn:maxLL-infectionpart-genIC} and \eqref{eqn:maxLL-failedinfectionpart-genIC}
together and taking log gives the result. 

Concavity again follows from the fact that $\log(1 - \exp(-x))$ is a concave function of $x$, and the fact that if any function $f(x)$ is a concave
function of $x$ then $f(\sum_i \theta_i)$ is jointly concave in $\theta$. $\square$

%%%%%%%%%%%%%%%%%%%%%%%%%%%%%%%%%%%%%%%%%%%%%%%%%%%%%%%%%%%%%%%%%%%%%%%%%%%%%%%%%%%%%%%%%%
\section{Markov Graphs and Causality}

\vspace{0.05in}
\subsection{Proof of Theorem \ref{thm:markov}}
We will show that $\mathbb{P}\left[T=t\right]$ can be written as a product of various factors where each
factor depends only on $t_{\V_i}$ for some $i\in V$. Given any vector $t$, for every $i \in V$ define the infection vectors
\begin{align*}
  x_i(\tau) = 
  \begin{cases}
    0 & \mbox{ if } 0 \leq \tau < t_i^{(1)} \\
    1 & \mbox{ if } t_i^{(1)} \leq \tau < t_i^{(2)} \\
    2 & \mbox{ if } \tau \geq t_i^{(2)}
  \end{cases}
\end{align*}
We can see that there is a one to one correspondence between valid time vectors $t$ and valid infection vectors $x$.
Using the above transformation, we can calculate the probability of a given time vector $t$ as follows:
\begin{align*}
  \mathbb{P}\left[T=t\right] &= \mathbb{P}\left[X=x\right] \\
	&= \mathbb{P}\left[X(0)=x(0)\right] \times \prod_{s=1}^{\infty} \mathbb{P}\left[X(s)=x(s) \; \middle \vert \; x(0:s-1)\right] \\
	&= \left(\prod_{i\in V}\mathbb{P}\left[X_i(0)=x_i(0)\right]\right) \times \prod_{s=1}^{\infty} \prod_{i\in V} \mathbb{P}\left[X_i(s)=x_i(s) \; \middle \vert \; x_{\V_i}(0:s-1)\right] \\
	&= \left(\prod_{i\in V}\mathbb{P}\left[X_i(0)=x_i(0)\right]\right) \times \prod_{i\in V} \prod_{s=1}^{\infty} \mathbb{P}\left[X_i(s)=x_i(s) \; \middle \vert \; x_{\V_i}(0:s-1)\right] \\
	&= \prod_{i\in V} f_i\left(t_{\V_i}\right)
\end{align*}
where $f_i(t_{\V_i}) = \mathbb{P}\left[X_i(0)=x_i(0)\right] \times \prod_{s=1}^{\infty} \mathbb{P}\left[X_i(s)=x_i(s) \; \middle \vert \; x_{\V_i}(0:s-1)\right]$.
\flushright $\square$

\end{document}